\documentclass[%
 reprint,
nofootinbib,
 amsmath,amssymb,
 aps,
]{revtex4-2}

\usepackage{amsmath,amssymb,amsthm,bm,mathrsfs}
\usepackage[dvipdfmx]{graphicx}
\usepackage{float}
\usepackage{subcaption}

\usepackage{multirow}
\usepackage{feynmp}

\usepackage{ascmac}
\usepackage{braket}
\usepackage{xcolor}
\usepackage{algorithm}
\usepackage{algpseudocode}
\usepackage{tikz}
\usepackage{graphicx}
\usetikzlibrary{decorations}
\usetikzlibrary{decorations.pathmorphing}
\usetikzlibrary{automata,positioning,arrows.meta,math, patterns}
\usepackage{array}
\usepackage{enumerate}
\usepackage{caption}
\usepackage{enumitem}
\usepackage{ulem}

\newcommand{\real}{\operatorname{Re}}
\newcommand{\imag}{\operatorname{Im}}

\theoremstyle{definition}
\newtheorem{theorem}{Theorem}[section]
\newtheorem{definition}[theorem]{Definition}
\newtheorem{lemma}[theorem]{Lemma}
\newtheorem{remark}[theorem]{Remark}
\newtheorem{proposition}[theorem]{Proposition}

\newtheorem{corollary}[theorem]{Corollary}

\makeatletter
\renewcommand{\ALG@name}{Algorithm}
\makeatother

\renewcommand{\emph}{\textit}

\usepackage{hyperref}
\begin{document}

\title{Ground state preparation in $(2+1)$-dimensional pure $\mathbb{Z}_2$ lattice gauge theory \\ via deterministic quantum imaginary time evolution}
\author{Minoru Sekiyama}%
\email[]{minorusekiyama(at)hep.phys.s.u-tokyo.ac.jp}
\affiliation{%
Department of Physics, Graduate School of Science, The University of Tokyo, 7-3-1 Hongo, Bunkyo-ku, Tokyo 113-0033, Japan
}%
\author{Lento Nagano}%
\email[]{lento.nagano(at)keio.jp}
\thanks{Current affiliation: Graduate School of Science and Technology, Keio University, Yokohama, Kanagawa
223-8522, Japan}
\affiliation{%
International Center for Elementary Particle Physics (ICEPP), The University of Tokyo, 7-3-1 Hongo, Bunkyo-ku, Tokyo 113-0033, Japan
}%

\date{\today}

\begin{abstract}
    In this paper, we apply the deterministic quantum imaginary time evolution (QITE) algorithm to obtain the ground state of a $(2+1)$-dimensional pure $\mathbb{Z}_2$ lattice gauge theory. 
    We first construct the set of Pauli operators commuting with Gauss's law constraints, generalizing a previous result. 
    This makes the deterministic QITE gauge-invariant and reduces both the measurement and gate costs significantly without adding extra algorithm errors in the QITE.
    Then, the classical numerical simulation of the deterministic QITE using tensor networks is performed, and the results are compared with the density matrix renormalization group (DMRG) to evaluate the accuracy of the algorithm.
Specifically, we investigate the coupling and system size dependence, and find that the deterministic QITE can achieve a relative error of less than $0.1\%$ up to a twelve-plaquette system and coupling values in a regime that we study.
Furthermore, the error dependence on the number of time steps is studied and discussed.
\end{abstract}

\maketitle
\section{Introduction}
Recently, there has been growing interest in Hamiltonian simulation of lattice gauge theories (LGTs), aiming to avoid the infamous sign problem in the conventional Monte-Carlo method.
However, naive classical implementations of Hamiltonian simulation often suffer from exponentially growing computational costs as system size increases.
To overcome this problem, quantum computing in the context of LGTs has been studied extensively in recent years since a seminal paper~\cite{Jordan:2012xnu} (see also e.g.~\cite{Bauer:2022hpo,DiMeglio:2023nsa,davoudi2025tasi,halimeh2025quantum} for recent reviews).

Among several applications, the preparation of the ground state and thermal state is important to probe properties of the theory, such as quantum phases. Imaginary time evolution (ITE) has been used for both purposes. However, implementation of (non-unitary) imaginary time evolution using a quantum computer is not straightforward, due to the unitary nature of quantum gates. Several attempts have been made in this context, including variational~\cite{McArdle:2019zrh,gomes2021adaptive,gacon2024variational} and deterministic~\cite{Motta:2019yya} methods that approximate the non-unitary ITE by unitary operators, an algorithm utilizing quantum signal processing~\cite {silva2023fragmented,chan2023simulating,zhang2025quantum}, an approach based on double-bracket flow~\cite{gluza2026double}, and probabilistic methods~\cite{Liu:2020zgv,lin2021real,kosugi2022imaginary}. 
Moreover, the ITE algorithm using continuous variables was proposed in~\cite{Yeter-Aydeniz:2021mol}.
The deterministic method proposed in Ref.~\cite{Motta:2019yya} (see also~\cite{yeter2020practical,gomes2020efficient,nishi2021implementation,huang2023efficient,melendez2025adaptive} for improvements and applications), which we call the deterministic quantum imaginary time evolution (QITE), approximates the imaginary time evolution by the real time evolution via a summation of Pauli operators, where the coefficients are determined by solving a linear equation.
The deterministic QITE algorithm has bounded errors if the support of the unitary operator is taken large enough. 
The drawback, however, is that it requires many measurements (tomography) to obtain the linear equation coefficients. Ref.~\cite{Sun:2020dzq} proposed a method to reduce the costs using the symmetry of the Hamiltonian, which was applied to the global symmetry in a spin model.
While the deterministic QITE and the reduction method can be applied to a general Hamiltonian, including that of LGTs, the computational cost depends on the details of the model considered. Therefore, resource estimation for a specific Hamiltonian would be important in practice.

In the context of LGTs, Refs.~\cite{Pedersen:2023asd,Davoudi:2022uzo}
obtained the thermal phase diagram of $(1+1)$-dimensional models using QITE along with thermal pure quantum states.
Moreover, the gauge-invariant implementation of thermal state preparation via the quantum minimally entangled typical thermal states (QMETTS) algorithm, which utilizes the QITE as a subroutine, was recently proposed for the $(1+1)$-dimensional~$\mathbb{Z}_2$ LGT~\cite{maeno2026efficient}.
Beyond one dimension, a $(2+1)$-dimensional $\mathbb{Z}_2$ LGT has been intensively studied in recent years as a simple toy model.
In particular, Ref.~\cite{wang2023symmetry} identified a set of Pauli operators commuting with Gauss's law (along with other symmetries in several models), and performed a numerical simulation of the variational QITE using an ansatz composed of the resulting Pauli set.
This not only makes the QITE gauge-invariant, but also reduces computational resources significantly~\footnote{Another related research direction is the continuous variable version of QITE proposed in Ref.~\cite{Yeter-Aydeniz:2021mol}. The authors further applied the method to a $(1+1)$-dimensional scalar field theory. Moreover, Ref.~\cite{ale2025simulating} proposed the QITE for a qubit-qumode hybrid system and applied it to the $(2+1)$-dimensional quantum electrodynamics.}.

In this paper, we apply the deterministic QITE algorithm (original qubit version) to obtain the ground state energy in the $(2+1)$-dimensional pure $\mathbb{Z}_2$ LGTs, and evaluate its accuracy and resources. 
First, we construct the set of Pauli operators in the linear equation compatible with Gauss's law, by generalizing the results in~\cite{wang2023symmetry} to arbitrary supports.
This preserves gauge symmetry in the QITE while significantly reducing resource requirements.
Then, we perform classical (and noiseless) numerical simulations for the deterministic QITE using the resulting Pauli operators and compare the results with those from the Suzuki-Trotterized ITE (without unitary approximation), and the density matrix renormalization group (DMRG)~\cite{PhysRevLett.69.2863}.
Specifically, we use the ladder-like geometry (two links in the vertical direction) and vary the system size (in the horizontal direction) and the coupling strength.  
We find that the deterministic QITE results are consistent with the DMRG results to within $0.1\%$, at least for small systems (up to twelve plaquettes) and coupling values in the regime that we study. 
The algorithmic errors are further investigated by studying the dependence on a time step size.

This paper is organized as follows. 
First, we review the general framework for the QITE algorithm~\cite{Motta:2019yya} and the reduction using symmetry~\cite{Sun:2020dzq} in Section~\ref{sec:QITE}.
Section~\ref{subsec:hamiltonian} defines the Hamiltonian of a pure $\mathbb{Z}_2$ LGT. Then, we explain the QITE and the symmetry reduction specifically in the $\mathbb{Z}_2$ LGT in Section~\ref{subsec:qite-and-reduction-z2lgt}.
With these setups, Section~\ref{sec:numerical-simulation} shows the results for the numerical simulations. 
Finally, we summarize our results and give some future directions in Section~\ref{sec:conclusion}.

\section{Quantum imaginary time evolution}
\label{sec:QITE}
\subsection{Ground state preparation via imaginary time evolution}
The imaginary time evolution (ITE) is defined as
\begin{equation}\label{eq:imaginary-time-solution}
    \ket{\psi_{\text{ITE}}(\tau)}=\exp(-\tau H)\ket{\psi_\text{init}}\,,
\end{equation}
where $\ket{\psi_\text{init}}$ is a chosen initial state.
Let $\ket{\Omega}$ be the ground state of a target Hamiltonian~\footnote{We assume that the ground state is non-degenerate.}, and suppose that we choose $\ket{\psi_{\text{init}}}$ such that $|\braket{\psi_{\text{init}}|\Omega}|\neq0$.
Then, the ITE in the infinite imaginary time limit gives the target ground state,
\begin{align}
    \lim_{\tau\to\infty}\ket{\bar{\psi}_{\text{ITE}}(\tau)}&=\ket{\Omega}\,,
    \\
    \ket{\bar{\psi}_{\text{ITE}}(\tau)} &= \frac{\ket{\psi_{\text{ITE}}(\tau)}}{\mathcal{N}(\ket{\psi_{\text{ITE}}(\tau)})}\,,
\end{align}
where $\mathcal{N}(\ket{\psi})$ is a normalization factor defined by~$\mathcal{N}(\ket{\psi}) = \sqrt{\braket{\psi|\psi}}$.
In practice, we truncate the imaginary time at~$\ket{\psi_{\text{ITE}}(\tau_{\text{max}})}$, which gives an approximation to the ground state.
\subsection{Quantum algorithm for imaginary time evolution}
We explain the implementation of the ITE using the quantum algorithm proposed in Ref.~\cite{Motta:2019yya}, which we call the deterministic \textit{quantum} imaginary time evolution (QITE).
While there are other families of QITE algorithms (e.g., variational, probabilistic), we only focus on the deterministic one, and we will denote the deterministic one simply as the QITE.
The QITE consists of two steps. 
First, the entire evolution is decomposed into small pieces by the standard Suzuki-Trotter decomposition.
More precisely, we divide the Hamiltonian as
\begin{align}
    H = \sum_{m=1}^{M} h^{(m)}\,,
\end{align}
where we assume that each term $h^{(m)}$ is at most $\kappa$-local with a constant $\kappa$.
With this notation, the first-order Suzuki-Trotter decomposition is given by
\begin{align}\label{eq:ST-decomomp-1st}
    e^{-\tau H} = \left(\prod_{m=1}^M e^{-\Delta \tau h^{(m)}}\right)^{N_{\text{step}}} + \mathcal{O}(\Delta \tau)\,,
\end{align}
where $\Delta \tau = \tau/N_{\text{step}}$ with a positive integer $N_{\text{step}}$. We sequentially apply~$MN_{\text{step}}$ terms in~\eqref{eq:ST-decomomp-1st} on the initial state~$\ket{\psi_{\text{init}}}$.

The next step is to approximate each term~$e^ {-\Delta \tau h^{(m)}}$ using unitary operators.
We call this step the \emph{unitary approximation}.
Let us denote an approximated (normalized) state at the $k$-th time step and $(m-1)$-th Hamiltonian decomposition by $\ket{\tilde{\Psi}^{(k,m-1)}_{\mathcal{P}}}$ ($k\in\{1,2,\cdots,N_{\text{step}}\}, m\in\{1,2,\cdots M\}$)~\footnote{
For $m=1$, we define $\ket{\tilde{\Psi}^{(k,0)}_{\mathcal{P}}}=\ket{\tilde{\Psi}^{(k-1,M)}_{\mathcal{P}}}$ for $k>1$ and $\ket{\tilde{\Psi}^{(1,0)}_{\mathcal{P}}}=\ket{\psi_\text{init}}$ for $k=1$.
Besides, the notation~$\mathcal{P}$ will be defined later.}. 
The target state in the next step which we want to obtain is given by
\begin{align}
    \ket{\bar{\psi}^{(k,m)}_{\text{ITE}}} &:=\frac{e^{-\Delta \tau h^{(m)}}\ket{\tilde{\Psi}^{(k,m-1)}_{\mathcal{P}}}}{\mathcal{N}(e^{-\Delta \tau h^{(m)}}\ket{\tilde{\Psi}^{(k,m-1)}_{\mathcal{P}}})} \,.
\end{align}
The QITE algorithm approximates the normalized state~$\ket{\bar{\psi}^{(k,m)}_{\text{ITE}}}$ by a (unitary) real-time evolution via an operator $A_\mathcal{P}$ defined as
\begin{align}\label{eq:QITE-approx-real-time-evol}
    \ket{{\Psi}^{(k,m)}_{\mathcal{P}}(\bm{a}^{(k,m)})} := e^{-i \Delta \tau A_\mathcal{P} (\bm{a}^{(k,m)})}\ket{\tilde{\Psi}^{(k,m-1)}_{\mathcal{P}}}\,,
\end{align}
where we define the unitary operator $A_\mathcal{P}$ as
\begin{align}
    A_\mathcal{P}(\bm{a}^{(k,m)}) = 
\sum_{\bm{\sigma}\in\mathcal{P}}
     a^{(k,m)}_{\bm{\sigma}} \bm{\sigma}\,,
\end{align}
where~$\mathcal{P}$, which we call the \textit{Pauli pool}, is a set of Pauli strings $\{I,X,Y,Z\}^{\otimes N_{\text{q}}}$ with $N_{\text{q}}$ being the total number of qubits. We will omit the Pauli pool dependence when it is clear.
The coefficients~$\bm{a}^{(k,m)}$ can be obtained by minimizing the distance between the two states,
\begin{align}
    \frac{1}{\Delta\tau^2}||\ket{\bar{\psi}^{(k,m)}_{\text{ITE}}} - \ket{{\Psi}^{(k,m)}_{\mathcal{P}}}||^2\,.
\end{align}
Expanding the minimization condition up to first order of $\Delta \tau$ gives the following linear equation for $\{a_{\bm{\sigma}}^{(k,m)}\}$.
\begin{align}
    \sum_{\bm{\sigma}'\in\mathcal{P}} S^{(k,m)}_{\bm{\sigma},\bm{\sigma}'} a_{\bm{\sigma}'}^{(k,m)} = b_{\bm{\sigma}}^{(k,m)}\,,\quad
    \bm{\sigma}\in\mathcal{P}\,.
    \label{eq:linear-eq-QITE}
\end{align}
The matrix elements are defined as (up to $\mathcal{O}(\Delta \tau^2)$ corrections)
\begin{align}
    \label{eq:linear-eq-coeffs-QITE}
\begin{cases}
    {S}^{(k,m)}_{\bm{\sigma},\bm{\sigma}'}=
2\real{\langle\bm{\sigma}\bm{\sigma}'\rangle}\,,
    \\
    {b}^{(k,m)}_{\bm{\sigma}} =
    -2\imag\langle\bm{\sigma} h^{(m)}\rangle/\sqrt{1-2\Delta\tau\langle h^{(m)}\rangle}\,,
\end{cases}
\end{align}
where we denote $\langle\mathcal{O}\rangle=\bra{\tilde\Psi^{(k,m-1)}_{\mathcal{P}}}\mathcal{O}\ket{\tilde\Psi^{(k,m-1)}_{\mathcal{P}}}$ (see Appendix~\ref{sec:derivation-QITE} for the derivation).
One can solve the linear equation~\eqref{eq:linear-eq-QITE} via a classical solver to have a solution denoted by~$\tilde{\bm{a}}^{(k,m)}$. We define the corresponding operator as $\tilde{A}^{(k,m)}_{\mathcal{P}}:= {A}_{\mathcal{P}}(\tilde{\bm{a}}^{(k,m)})$ and denote the resulting state by
\begin{align}
    \ket{\tilde{\Psi}^{(k,m)}_\mathcal{P}}
    &:= \ket{{\Psi}^{(k,m)}_\mathcal{P}(\tilde{\bm{a}}^{(k,m)}_{\mathcal{P}})}
    \\
    &= e^{-i \Delta \tau \tilde{A}^{(k,m)}_\mathcal{P}}\ket{{\tilde\Psi}^{(k,m-1)}_{\mathcal{P}}}
    \,.
\end{align}
In order to obtain the resulting state, we implement the Suzuki-Trotter decomposition of $\tilde{A}^{(k,m)}_{\mathcal{P}}$, which requires~$\mathcal{O}(|\mathcal{P}|)$ gates in the worst case. 
We repeat this procedure $\ket{\tilde{\Psi}^{(k,m-1)}_\mathcal{P}}\to \ket{\tilde{\Psi}^{(k,m)}_\mathcal{P}}$ for each Hamiltonian term $h^{(m)}$ and each Suzuki-Trotter step.
The pseudocode for the QITE algorithm is given in Algorithm~\ref{alg:QITE-trotter}.

In the numerical simulation in Section~\ref{sec:numerical-simulation}, we use the second-order Suzuki-Trotter decomposition,
\begin{align}
\left(\prod_{m=1}^M e^{-\Delta \tau h^{(m)}/2}\prod_{m=M}^1 e^{-\Delta \tau h^{(m)}/2}\right)^{N_{\text{step}}}\,,
\end{align}
which can be implemented similarly by doubling the terms~$h^{(m)}$ and replacing the time step with~$\Delta\tau/2$.
This suppresses Suzuki-Trotter error in the simulation and makes the QITE-specific error (unitary approximation error) more prominent, which is the main interest of this paper.

\begin{algorithm}[H]
\caption{QITE algorithm}
\label{alg:QITE-trotter}
\begin{algorithmic}[1]

\State $\Delta\tau \gets \tau_{\max}/N_{\text{step}}$
\State state initialization $\ket{\tilde{\Psi}^{(k=0,m=M)}}\gets \ket{\psi_{\text{init}}}$
\For{$k = 1$ \textbf{to} $N_{\text{step}}$} 
    \Comment loop for time steps
     \State $\ket{\tilde{\Psi}^{(k,m=0)}}\gets \ket{\tilde{\Psi}^{(k-1,M)}}$
    \For{$m = 1$ \textbf{to} $M$} \Comment loop for Hamiltonian terms
        \State construct $\bm{S}^{(k,m)},\bm{b}^{(k,m)}$ via Eq.~\eqref{eq:linear-eq-coeffs-QITE} 
        \State find $\tilde{\bm{a}}^{(k,m)}$ by solving Eq.~\eqref{eq:linear-eq-QITE}

        \State $\tilde{A}^{(k,m)} \gets \sum_{\bm{\sigma}}\tilde{a}^{(k,m)}_{\bm{\sigma}} \bm{\sigma}$, 
        \State 
        $\ket{\tilde{\Psi}^{(k,m)}} \gets e^{-i\,\Delta\tau\,\tilde{A}^{(k,m)}}\ket{\tilde{\Psi}^{(k,m-1)}}$
    \EndFor
\EndFor
\end{algorithmic}
\end{algorithm}

\subsection{Reduction of Pauli pool}
\label{subsec:reduction-pauli-pool}
In this subsection, we review a general method to reduce the dimension of the Pauli pool according to Ref.~\cite{Sun:2020dzq}.
The proofs of the propositions given in this section are provided in Appendix~\ref{app:proof-pauli-pool-reduction}.
In the following, we introduce the support and weight of the Pauli string~$\bm{\sigma}$, which will be used to represent the effective size of the Pauli pool. Essentially, weight counts the number of (non-identity) Pauli operators, and support is a region on which the Pauli string acts nontrivially. The formal definition is given as follows.

\begin{definition}[Support and weight of Pauli string]
    Let $\bm{\sigma}$ be a Pauli string expressed as $\bm{\sigma}=\sigma_{q_1}^{\mu_1}\otimes\cdots \otimes\sigma^{\mu_k}_{q_k}$ where $\mathcal{Q}=\{q_i\}$ is a set of qubits and 
    $\mu_i \in \{I,X,Y,Z\}$ with~$\sigma^I=I$. The \textit{support} of the Pauli string is defined as a set of qubits on which $\bm{\sigma}$ acts nontrivially, i.e.,
    \begin{equation}
    \text{supp}(\bm{\sigma}) = \{q_i \in \mathcal{Q}\mid \mu_i \neq I\}\,.
    \end{equation}
    We also define the \textit{weight} of the Pauli string $\bm{\sigma}$ by the number of qubits in $\text{supp}(\bm{\sigma})$, that is, 
    \begin{equation}
        \text{wt}(\bm{\sigma}) := |\text{supp}(\bm{\sigma}) | \,.
    \end{equation}
\end{definition}

To obtain the coefficients in the linear equation~\eqref{eq:linear-eq-QITE} using the Pauli pool $\mathcal{P}$
we need to evaluate $\mathcal{O}(|\mathcal{P}|^2)$ terms defined in Eq.~\eqref{eq:linear-eq-coeffs-QITE},
which increases the total number of measurements. 
For example, if we choose the Pauli pool as $\mathcal{P}=\{I,X,Y,Z\}^{\otimes q}$, then~$\mathcal{O}(|\mathcal{P}|^2)\sim 2^{\mathcal{O}(q)}$.
Furthermore, the implementation of $e^{-i\tilde{A}_\mathcal{P}\Delta \tau}$ involves the Suzuki-Trotter decomposition of~$\mathcal{O}(|\mathcal{P}|)$ terms, which requires gate costs of the same order in the worst case.
Although it was shown in Ref.~\cite{Motta:2019yya} that the size of the support is bounded by~$\mathcal{O}(C^2)$, where $C$ is a correlation length, for a
nearest-neighbor local Hamiltonian on a two-dimensional square lattice, the exponential increase in measurement cost is a bottleneck of this algorithm.
To reduce this cost, one can use the following three propositions, obtained in the previous work~\cite{Sun:2020dzq}. 

The first proposition utilizes the reality of the coefficients in the linear equation~\eqref{eq:linear-eq-coeffs-QITE}.
\begin{definition}
We denote the number of $Y$ operators in a Pauli string~$\bm{\sigma}\in\mathcal{P}$ by $\text{num}_Y(\bm{\sigma})$. 
Additionally, we define the Pauli pool whose elements consist of Pauli strings with an odd number of $Y$ as
\begin{equation}
 \mathcal{P}_{\text{odd}}:=
        \{\bm{\sigma}\in \mathcal{P} \mid \text{num}_Y(\bm{\sigma}) \equiv1 \pmod 2 \}  \,. 
\end{equation}
\end{definition}

\begin{proposition}[Reduction via reality condition]
\label{prop:reduction-by-reality}
    Suppose that all elements of the matrix (or vector) representation of $\ket{\tilde{\Psi}^{(k,m-1)}_{\mathcal{P}}}$ and $h^{(m)}$ in the computational basis are real. 
Then the solution $\tilde{\bm{a}}$ of the linear equation~\eqref{eq:linear-eq-QITE} for the Pauli pool~$\mathcal{P}$ is derived from 
the reduced Pauli pool $\mathcal{P}_{\text{odd}}$.
Moreover, the real-time evolutions obtained from both Pauli pools are equivalent, i.e.,
\begin{equation}\label{eq:equivalence-P-and-Podd}
 \ket{\tilde{\Psi}^{(k,m)}_{\mathcal{P}}}
=
\ket{\tilde{\Psi}^{(k,m)}_{\mathcal{P}_{\text{odd}}}}   \,.
\end{equation}
\end{proposition}
This proposition implies two important cost reductions.
First, the evaluation of matrix elements in Eq.~\eqref{eq:linear-eq-coeffs-QITE} is needed only for the reduced Pauli pool $\mathcal{P}_{\text{odd}}$.
On top of that, $e^{-i\Delta \tau A_{\mathcal{P}}}$ can be replaced with $e^{-i\Delta \tau A_{\mathcal{P}_{\text{odd}}}}$ while keeping the accuracy unchanged (see Eq.~\eqref{eq:equivalence-P-and-Podd}), which requires fewer gates in general.

The second reduction method uses symmetry.
\begin{definition}[Symmetrized Pauli pool]\label{def:reduction-symmetry}
   Let~$\mathcal{S}$ be a symmetry group. For a given Pauli pool $\mathcal{P}$, we define the symmetrized Pauli pool as 
    \begin{equation}\label{eq:def-symm-pauli}
    \mathcal{P}_\mathcal{S} := \{\bm{\sigma}\in \mathcal{P} \mid \bm{\sigma} = s \bm{\sigma} s^\dagger, \forall s\in \mathcal{S}\} \,.   
    \end{equation}
\end{definition}
\begin{proposition}[Reduction via symmetry]
\label{prop:reduction-symmetry}
Let us denote the model Hamiltonian as $H$ and a symmetry group as~$\mathcal{S}$.
We then assume the following conditions.
\begin{itemize}
    \item (symmetry of initial state) The initial state satisfies $s\ket{\psi_{\text{init}}}=\ket{\psi_{\text{init}}}$ for all~$s\in \mathcal{S}$.
    \item (symmetry of Hamiltonian) The Hamiltonian satisfies $[H,s]=0$ for all $s \in \mathcal{S}$
    \item (symmetric decomposition) Each term in the Suzuki-Trotter decomposition satisfies $[h^{(m)},s]=0$ for all $m=1,\ldots,M$ and $s\in \mathcal{S}$.
\end{itemize}
With these conditions, 
the solution of the equation~\eqref{eq:linear-eq-QITE} for the Pauli pool~$\mathcal{P}$ is derived from 
the Pauli pool~${\mathcal{P}_\mathcal{S}}$.
Moreover, the real-time evolutions obtained from both Pauli pools are equivalent, i.e.,
\begin{equation}
\ket{\tilde{\Psi}^{(k,m)}_{\mathcal{P}}} = \ket{\tilde{\Psi}^{(k,m)}_{\mathcal{P}_\mathcal{S}}}\,.
\end{equation}
\end{proposition}
As explained after Prop.~\ref{prop:reduction-by-reality}
, Prop.~\ref{prop:reduction-symmetry} also leads to reductions both in measurement and gate costs. Moreover, the resulting state automatically preserves the symmetry, i.e.,
\begin{equation}
    s e^{-i A_{\mathcal{P}_\mathcal{S}}\Delta \tau}\ket{\psi_{\text{init}}}=e^{-iA_{\mathcal{P}_\mathcal{S}}\Delta \tau}\ket{\psi_{\text{init}}}\,,
\end{equation}
holds
for all $s\in\mathcal{S}$~\footnote{If there is any error in evaluating the coefficients~\eqref{eq:linear-eq-coeffs-QITE}, those correspond to non-symmetric Paulis $\bm{\sigma}\notin \mathcal{P}_\mathcal{S}$ become zero, even if one solves the linear equation~\eqref{eq:linear-eq-QITE} for the entire Pauli pool~$\mathcal{P}$. However, once some errors occur in the evaluation, they may turn on these elements and break the symmetry.
In this sense, Prop.~\ref{prop:reduction-symmetry} makes the symmetry in the QITE more robust to the errors.}. 

One can further reduce the size of the Pauli pool by considering the quotient group under the action of $\mathcal{S}$.
This can be formulated as follows.
\begin{definition}[Quotient Pauli pool]
   Let $\mathcal{S}$ be a symmetry group, which is composed of Pauli strings.
   We introduce a binary relation~$\overset{\mathcal{S}}{\sim}$ defined by
   \begin{equation}
       \bm{\sigma}\overset{\mathcal{S}}{\sim}\bm{\sigma}'
       \Leftrightarrow
       \exists s \in \mathcal{S}, \,\text{such that} \,\bm{\sigma}=\pm\bm{\sigma}'s\,.
   \end{equation}
   Then, the \emph{quotient Pauli pool} is defined as
   \begin{equation}
       \mathcal{P}_{\mathcal{S}}/\mathcal{S} := \mathcal{P}_{\mathcal{S}}/\overset{\mathcal{S}}{\sim}\,.
   \end{equation}
Each element of the quotient group is represented by its representative as $[\bm{\sigma}]_\mathcal{S}:= \{\bm{\sigma}' \in \mathcal{P}_\mathcal{S} \mid \exists s \in \mathcal{S},\bm{\sigma}=\pm\bm{\sigma}'s\}$.
\end{definition}
Below, we will implicitly fix a representative for each element in $\mathcal{P}_\mathcal{S}/\mathcal{S}$.

\begin{proposition}[Reduction via quotient group]
Suppose that all the assumptions of Prop.~\ref{prop:reduction-symmetry} hold. 
Then, the solution of the linear equation~\eqref{eq:linear-eq-QITE} for the symmetrized Pauli pool~$\mathcal{P}_\mathcal{S}$ is derived from that of 
the quotient Pauli pool $\mathcal{P}_{\mathcal{S}}/\mathcal{S}$.
Moreover, the real-time evolutions obtained from both Pauli pools are equivalent, i.e.,
\begin{align}
    \ket{\tilde{\Psi}^{(m)}_{\mathcal{P}_\mathcal{S}}} = \ket{\tilde{\Psi}^{(m)}_{\mathcal{P}_\mathcal{S}/\mathcal{S}}}\,.
\end{align}
\end{proposition}
This implies that it is sufficient to consider the representatives in each equivalent class.

Finally, we can combine those propositions into a single statement~\footnote{In general, even if $\mathcal{P}_1, \mathcal{P}_2$ satisfies $\ket{\tilde{\Psi}^{(m)}_{\mathcal{P}}} = \ket{\tilde{\Psi}^{(m)}_{\mathcal{P}_i}}$ ($i=1,2$), it does not necessarily hold $\ket{\tilde{\Psi}^{(m)}_{\mathcal{P}}} = \ket{\tilde{\Psi}^{(m)}_{\mathcal{P}_1 \cap \mathcal{P}_2}}$. However, at least in Prop.~\ref{prop:reduction-combined}, one can show that the intersection gives the correct reduction. See Appendix~\ref{app:proof-pauli-pool-reduction} for the details. }.

\begin{proposition}
\label{prop:reduction-combined}
     Under the assumptions of the previous propositions, the original Pauli pool $\mathcal{P}$ is reduced to the intersection of the odd Pauli pool $\mathcal{P}_\text{odd}$ and the symmetrized Pauli pool $\mathcal{P}_{\mathcal{S}}$
\begin{align}
    \ket{\tilde{\Psi}^{(m)}_{\mathcal{P}}} = \ket{\tilde{\Psi}^{(m)}_{\mathcal{P}_\text{odd}\cap\mathcal{P}_\mathcal{S}}}\,.
\end{align}
    In addition, the original Pauli pool $\mathcal{P}$ is further reduced to the intersection of the odd Pauli pool $\mathcal{P}_\text{odd}$ and the quotient Pauli pool of $\mathcal{S}$-action $\mathcal{P}_{\mathcal{S}}/\mathcal{S}$ as
\begin{align}
    \ket{\tilde{\Psi}^{(m)}_{\mathcal{P}}} = \ket{\tilde{\Psi}^{(m)}_{\mathcal{P}_\text{odd}\cap(\mathcal{P}_\mathcal{S}/\mathcal{S})}}\,.
\end{align}
\end{proposition}

We will apply these propositions to the LGT model in the next section.

\section{QITE in $\mathbb{Z}_2$ lattice gauge theory}
\label{sec:qite-z2LGT}
\subsection{Hamiltonian}
\label{subsec:hamiltonian}
We consider a pure $\mathbb{Z}_{2}$ gauge theory on a two-dimensional square lattice~\cite{Wegner:1971app,kogut1979}.
The $\mathbb{Z}_{2}$ gauge theory is the simplest toy model for two-dimensional LGTs; no truncation is required for the gauge variables since it only takes discrete values.
However, it still possesses similar interesting properties in a full-fledged LGT; it exhibits confinement/deconfinement, while its extension with dynamical matter also exhibits string breaking. For this reason, this model with and without matter has been intensively studied in recent years using quantum computers and tensor networks (see e.g.,~\cite{borla2025string,xu2025string,Sugihara:2005cf,Tagliacozzo:2010vk,Liu:2013nsa,Wu:2025aly,Xu:2025ean,Cochran:2024rwe,Yamamoto:2020eqi,alexandrou2025realizing,azad2025barren,2026arXiv260407435X,emonts2023finding,irmejs2023quantum,lumia2022two,cobos2025real,mueller2025quantum,ding2021digital,wiese2013ultracold,di2025roughening,homeier2023realistic,gonzalez2020robust,sukeno2023measurement,borla2024deconfined,xu2025critical}).

A square lattice we consider is composed of~$N_{\text{link}}$ links where the gauge fields live and~$N_{\text{site}}$ sites. 
The $\mathbb{Z}_2$ link variables are directly mapped to qubits, that is, $N_{\text{q}}=N_{\text{link}}$.
We can put static fermionic matter at each site, but we will not do this in this paper. 
Each site is specified by a Cartesian coordinate~$\bm{n}=(n_x,n_y)$ for $n_\mu \in\{0,1,\cdots,N_{\text{site}}^\mu-1\}$, where~$N_{\text{site}}^\mu\,(\mu=x,y)$ is the number of sites in each direction.
An operator~$O$ acting on the link between~$\bm{n}$ and~$\bm{n}+\bm{e_{\mu}}\, (\mu =x, y)$ is denoted by~$O_{\bm{n}, \mu}$ where $\bm{e_{\mu}}$ is a unit vector in the $\mu$-th direction.
The Hamiltonian is given by~\footnote{The convention of the coupling in this article is different from the one used in the standard Kogut-Susskind Hamiltonian~\cite{PhysRevD.11.395}. The small coupling in our convention corresponds to the strong coupling (confining regime) in the standard one.}
\begin{align}
H&=
-\sum_{\bm{n}
}
\left(
X_{\bm{n},x}
+
X_{\bm{n},y}
\right)
-\lambda\sum_{\bm{n}}
W_{\bm{n}}
\,,
\\
W_{\bm{n}}& =
Z_{\bm{n},x}
Z_{\bm{n}+\bm{e_x},y}
Z_{\bm{n}+\bm{e_y},x}
Z_{\bm{n},y}\,.
\end{align}
The first term is called the electric term and the second is called the magnetic term.
The physical states must satisfy the following conditions, called Gauss's law constraints,
\begin{align}  
    g_{\bm{n}}\ket{\text{phys}} 
    &= (-1)^{q_{\bm{n}}} \ket{\text{phys}}\,,\
\\
 g_{\bm{n}}
 &=
 X_{\bm{n}-\bm{e_x},\bm{e_x}}
 X_{\bm{n},\bm{e_x}}
 X_{\bm{n}-\bm{e_y},\bm{e_y}}
X_{\bm{n},\bm{e_y}}\label{eq:gauss-law}
\end{align}
  where $q_{\bm{n}}$ specifies a static charge on the vertex~$\bm{n}$ (see Fig.~\ref{fig:hamiltonian-terms}). We set~$q_{\bm{n}}=0$ in this paper, but the following argument can be applied to any values of $q_{\bm{n}}$.
Additionally, we impose open boundary conditions on all boundaries.
Specifically, if a site $\bm{n}$ is at a boundary and a link~$(\bm{n}\pm\bm{e_\mu},\mu)$ does not exist, we define~$X_{\bm{n}\pm\bm{e_\mu},\mu}=1$.
Then, the corresponding Gauss's law operator~\eqref{eq:gauss-law} becomes a three (two) body operator at boundaries (corners) of the rectangular lattice.

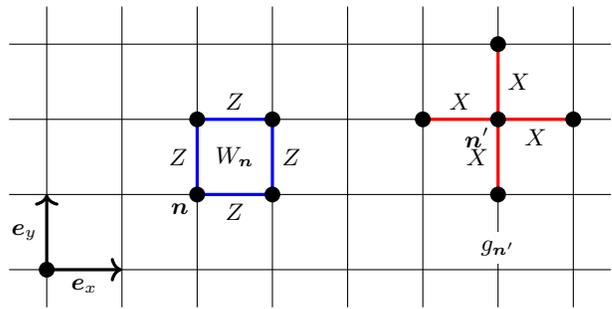
\begin{figure}[!bt]
\centering
    \begin{tikzpicture}

\foreach \y in{0,...,3}{
\draw[thin] (-0.5,\y)--(7.5,\y);
}
\foreach \x in{0,...,7}{
\draw[thin] (\x,-0.5)--(\x,3.5);
}

\coordinate (qO) at (0,0);\fill (qO) circle (3pt);
\coordinate (qX) at (1,0);
\coordinate (qY) at (0,1);

\draw[very thick, ->] (qO)--node[below]{$\bm{e}_x$}(qX);
\draw[very thick, ->] (qO)--node[left]{$\bm{e}_y$}(qY);

\coordinate (q00) at (2,1);
\coordinate (q01) at (2,2);
\coordinate (q10) at (3,1);
\coordinate (q11) at (3,2);

\draw [very thick, color=blue](q00)  --node[left, color=black] {$Z$}(q01)--node[above, color=black] {$Z$}(q11) --node[right, color=black] {$Z$}(q10)--node[below, color=black] {$Z$}(q00);
\node[] at (2.5,1.5) {$W_{\bm{n}}$};
\node[below left] at (q00) {$\bm{n}$};

\fill (q00) circle (3pt);
\fill (q01) circle (3pt);
\fill (q10) circle (3pt);
\fill (q11) circle (3pt);

\coordinate (q001) at (5,2);
\coordinate (q010) at (6,1);
\coordinate (q011) at (6,2);
\coordinate (q021) at (7,2);
\coordinate (q012) at (6,3);

\draw [very thick, color=red](q001)--node[above, color=black] {$X$}(q011) --node[below, color=black] {$X$}(q021);
\draw [very thick, color=red](q010)--node[left, color=black]{$X$}(q011)--node[right, color=black]{$X$}(q012);

\node[below left] at (q011) {$\bm{n}^\prime$};
\node[below,fill=white] at (6,0.5) {$g_{\bm{n}^\prime}$};

\fill (q001) circle (3pt);
\fill (q010) circle (3pt);
\fill (q011) circle (3pt);
\fill (q021) circle (3pt);
\fill (q012) circle (3pt);

\end{tikzpicture}
\caption{The magnetic term and Gauss's law in our convention. Blue links denote a magnetic term. Red links denote a Gauss's law operator.}
\label{fig:hamiltonian-terms}
\end{figure}

\subsection{QITE and Pauli pool reduction}
\label{subsec:qite-and-reduction-z2lgt}
In this subsection, we will apply the general reduction methods explained in Section~\ref{subsec:reduction-pauli-pool} to the $\mathbb{Z}_2$ LGT and show that it reduces the size of a Pauli pool. This was done in~\cite{wang2023symmetry} for a fixed support (one plaquette), but we give a general formula for the size of the Pauli pool for an arbitrary support. On top of that, we also apply the reduction via the quotient group (i.e., Prop.~\ref{prop:quotient}). 

Let us first explain the setups. We decompose the Hamiltonian into subterms as~$\{h^{(m)}\}_m=\{X_{\bm{n},\bm{e}_\mu}\}_{\bm{n},\mu} \cup \{W_{\bm{n}}\}_{\bm{n}}$ and apply the second-order Suzuki-Trotter decomposition to~$e^{-\tau H}$. 
We also use the ground state of the electric term, denoted by~$\ket{\Omega_{0}}:=\ket{+}^{\otimes N_{\text{link}}}$, as an initial state. Finally, the Pauli pool acting on a set of qubits (links) $\mathcal{D}:=\{q_1,\cdots, q_D\}$ is defined by
\begin{equation}\label{eq:pauli-group-in-domain}
    \mathcal{P}[\mathcal{D}]:=\{\sigma^{\mu_1}_{q_1}\otimes \cdots \otimes \sigma^{\mu_D}_{q_D}, \mu_i \in \{I,X,Y,Z\} \} \,.
\end{equation}

We call the region $\mathcal{D}$ the support of the Pauli pool $\mathcal{P}[\mathcal{D}]$.
We will omit $\mathcal{D}$-dependence when the support is fixed and obvious.

Under these setups, we employ the propositions given in Section~\ref{subsec:reduction-pauli-pool}.
First of all, one can show that the assumptions in Prop.~\ref{prop:reduction-by-reality} are satisfied for our settings. Therefore, the Pauli pool $\mathcal{P}$ is reduced to $\mathcal{P}_{\text{odd}}$.

We can further reduce the resulting Pauli pool using symmetry via Prop.~\ref{prop:reduction-symmetry}.
The symmetry group we consider is generated by Gauss's law as
\begin{equation}\label{eq:gauss-law-group}
 \mathcal{G}[\mathcal{D}]:=\Braket{g_{s_1},g_{s_2},\cdots, g_{s_{n_{\text{site}}(\mathcal{D})}}} \,,  
\end{equation}
where $s_i$ is a site in the support $\mathcal{D}$, $n_{\text{site}}(\mathcal{D})$ is the number of sites involved in the support~$\mathcal{D}$, and $g_n$ is defined by~Eq.~\eqref{eq:gauss-law}.
Now, we split elements in the Pauli pool as~$\bm{\sigma}=\bm{\sigma}_X \otimes \bm{\sigma}_{\bar{X}}$, where~$\bm{\sigma}_X$ is a tensor product of~$X$ and~$I$ operators, and~$\bm{\sigma}_{\bar{X}}$ is a product of~$Y$ and~$Z$ operators. 

To identify the condition under which $[\bm{\sigma},g]=0$, we follow a similar argument used to find the detectable errors with syndromes in the toric code~\cite{Kitaev:1997wr,Dennis:2001nw}.
Suppose that $\text{supp}(\bm{\sigma}_{\bar{X}})=\{q_1,\cdots ,q_{D'}\}$. We identify the set of qubits as a set of links where the qubits live. Then, if $\text{supp}(\bm{\sigma}_{\bar{X}})$ has a boundary vertex $s$ as Fig.~\ref{fig:gauss-law-explanation}, the corresponding Gauss's law operator $g_s$ overlaps with the support of $\bm{\sigma}_{\bar{X}}$ an odd number of times, resulting in $[\bm{\sigma},g_s]=[\bm{\sigma}_{\bar{X}},g_s]\neq 0$. On the other hand, if the support forms the closed loop, all the Gauss's law operators overlap with the $\bm{\sigma}_{\bar{X}}$ support an even number of times, and therefore $\bm{\sigma}_{\bar{X}}$ commutes with all of them.
To recap, we have found that
\begin{equation}
    [\bm{\sigma},g]=0, \forall g\in\mathcal{G}
    \Leftrightarrow
    \text{supp}(\bm{\sigma}_{\bar{X}}) \in \Gamma\,,
\end{equation}
where $\Gamma$ is a set of closed loops (which can be disjoint). 
In other words, the condition $[\bm{\sigma},g]=0$ enforces that~$\bm{\sigma}_{\bar{X}}$ should be gauge invariant, and therefore its support is geometrically similar to Wilson loops in the pure theory.
Therefore, the symmetrized Pauli group defined by Eq.~\eqref{eq:def-symm-pauli} is given as
\begin{equation}
 \mathcal{P}_{\mathcal{G}}=\{\bm{\sigma}\in\mathcal{P}\mid \text{supp}(\bm{\sigma}_{\bar{X}}) \in \Gamma\}   \,.
\end{equation}

\begin{figure}[!b]
\centering
\scalebox{0.7}{
    \begin{tikzpicture}

\foreach \y in{0,...,3}{
\draw[thin] (-3.5,\y)--(8.5,\y);
}
\foreach \x in{-3,...,8}{
\draw[thin] (\x,-0.5)--(\x,3.5);
}

\draw [line width =4.5pt, color=red!60](-2,3)--(-2,2) --(-2,1);
\draw [line width =4.5pt, color=red!60](-3,2)--(-2,2) --(-1,2);

\draw [line width =4.5pt, color=red!60](2,2)--(2,1) --(2,0);
\draw [line width =4.5pt, color=red!60](1,1)--(2,1) --(3,1);

\draw [line width =4.5pt, color=red!60](5,2)--(5,1) --(5,0);
\draw [line width =4.5pt, color=red!60](4,1)--(5,1) --(6,1);

\draw[line width=2.pt,color=blue] (-2,2)--(0,2)--(0,0)--(2,0)--(2,1)--(7,1)--(7,2)--(6,2)--(6,3)--(5,3)--(5,1);

\fill[black] (-2,2) circle (4.2pt);
\fill[white] (-2,2) circle (3pt);
\fill[black] (2,1) circle (4.2pt);
\fill[black] (5,1) circle (4.2pt);
\fill[white] (5,1) circle (3pt);

\end{tikzpicture}
}
\caption{
Gauss's law constraints on Pauli pool elements. Blue links represent the support of~${\bm{\sigma}}_{\bar{X}}$, and red links represent the support of the Gauss's law operator. At a black vertex (site), the number of overlapping links is even, and thus its eigenvalue does not change.  On the other hand, Gauss's law at white vertices overlaps with the blue support an odd number of times, flipping the sign of Gauss's law eigenvalues.}
\label{fig:gauss-law-explanation}
\end{figure}

Using Prop~\ref{prop:reduction-symmetry}, one can show that the Pauli pool $\mathcal{P}$ is reduced to $\mathcal{P}_\mathcal{G}$. As mentioned after Prop.~\ref{prop:reduction-symmetry}, this not only saves the measurement and gate cost, but also makes the resulting state gauge invariant, i.e., preserves all the eigenvalues of generators~$g_n$.

Combining the above two results, the reduction from the full pool $\mathcal{P}$ to 
\begin{align}
        \mathcal{P}_{\mathcal{G},\text{odd}} &:= \mathcal{P}_{\mathcal{G}} \cap \mathcal{P}_{\text{odd}}
        \\
        &=
        \{\bm{\sigma}\in\mathcal{P}\mid \text{supp}(\bm{\sigma}_{\bar{X}}) \in \Gamma,
        \notag\\ &\qquad
        \land \text{num}_Y(\bm{\sigma}_{\bar{X}}) \equiv 1 \pmod 2\}\,.
\end{align}
The number of elements (denoted by $|\bullet|$) in the reduced Pauli pool~$\mathcal{P}_{\mathcal{G},\text{odd}}$ is given by
\begin{align}\label{eq:num-elements-reduced-pool}
        |\mathcal{P}_{\mathcal{G},\text{odd}}[\mathcal{D}]|={2^{n_{\text{link}}(\mathcal{D})-1}\times(2^{n_{\text{plaq}}(\mathcal{D})}-1)}\,,
\end{align}
where we define the number of links (plaquettes) involved in the support $\mathcal{D}$ by $n_{\text{link (plaq)}}(\mathcal{D})$. 
See Appendix~\ref{app:pauli-pool-proof-z2lgt} for a proof.
Explicit values of the right-hand side of Eq.~\eqref{eq:num-elements-reduced-pool} for concrete examples are given in Tab.~\ref{tab:reduced-pauli-pool-z2lgt}. We observe a significant reduction using these methods, which saves the number of measurements for Eq.~\eqref{eq:linear-eq-coeffs-QITE} and the number of gates in the resulting real-time evolution~\eqref{eq:QITE-approx-real-time-evol}.
The formula~\eqref{eq:num-elements-reduced-pool} for a one-plaquette support (the first row in Tab.~\ref {tab:reduced-pauli-pool-z2lgt}) agrees with what was obtained in~\cite{wang2023symmetry}.

Finally, let us consider the Pauli pool reduction via the quotient group (Prop.~\ref{prop:reduction-combined}).
For this purpose, we classify the Gauss's law operators appearing in Eq.~\eqref{eq:gauss-law-group} as follows.
If all four $X$ operators in an operator $g_n$ are included in the support $\mathcal{D}$, then this operator is called a \textit{bulk} Gauss's law operator, and otherwise referred to as a \textit{boundary} Gauss's law operator (see Fig.~\ref{fig:bulk-gauss-law}). On top of that, the number of bulk Gauss's law operators for the support $\mathcal{D}$ is denoted by $n_{\mathcal{G}}(\mathcal{D})$, which is equal to the number of bulk sites in the region $\mathcal{D}$. 
With these definitions, the number of elements in the quotient group is given by 
\begin{align}
|\mathcal{P}_{\mathcal{G},\text{odd}}/\mathcal{G}|
&=\frac{|\mathcal{P}_{\mathcal{G},\text{odd}}|}{2^{n_{\mathcal{G}}(\mathcal{D})}}
\\
&=\frac{2^{n_\text{link}(\mathcal{D})-1}\times(2^{n_\text{plaq}(\mathcal{D})}-1)}{2^{n_{\mathcal{G}}(\mathcal{D})}}\,.
\end{align}
See Appendix~\ref{app:pauli-pool-proof-z2lgt} for a proof.
The explicit value of this formula for several supports are given in Tab.~\ref{tab:reduced-pauli-pool-z2lgt}. 
One can see that the reduction via reality and symmetry condition is significant (e.g. $255\to8$ in the first row). On the other hand, for the regions $\mathcal{D}$ in this table, $n_\mathcal{G}(\mathcal{D})$ is at most one, and thus the reduction via quotient $\mathcal{P}_{\mathcal{G},\textrm{odd}} \to \mathcal{P}_{\mathcal{G},\textrm{odd}}/\mathcal{G}$ is at most two
\footnote{
The reduction factor using the reality and symmetry conditions is 
$|\mathcal{P}_{\mathcal{G},\textrm{odd}}|/|\mathcal{P}|\sim 2^{-n_{\textrm{site}}(\mathcal{D})}$,
where $n_{\text{site}}(\mathcal{D})$ denotes the number of sites included in the region $\mathcal{D}$.
The quotient group method further reduces the cost by
$|\mathcal{P}_{\mathcal{G},\textrm{odd}}/\mathcal{G}|/|\mathcal{P}_{\mathcal{G},\textrm{odd}}|=2^{-n_{\mathcal{G}}(\mathcal{D})}$. Since $n_{\text{site}}(\mathcal{D})>n_{\mathcal{G}}(\mathcal{D})$ holds, the symmetry and reality reduction is more significant than the quotient one when the region $\mathcal{D}$ is sufficiently large.}.
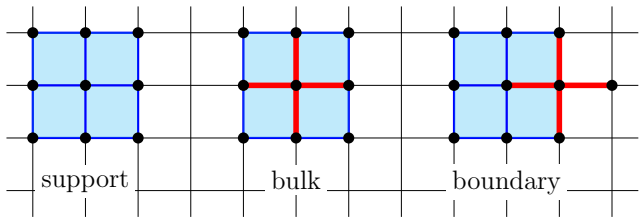
\begin{figure}[!t]
\centering
\scalebox{0.7}{
    \begin{tikzpicture}

\foreach \y in{0,...,3}{
\draw[thin] (-3.5,\y)--(8.5,\y);
}
\foreach \x in{-3,...,8}{
\draw[thin] (\x,-0.5)--(\x,3.5);
}

\coordinate (qm00) at (-3,1);
\coordinate (qm01) at (-3,2);
\coordinate (qm10) at (-2,1);
\coordinate (qm11) at (-2,2);
\coordinate (qm20) at (-3,3);
\coordinate (qm21) at (-2,3);
\coordinate (qm02) at (-1,1);
\coordinate (qm12) at (-1,2);
\coordinate (qm22) at (-1,3);

\draw [very thick, color=blue](qm00)  --(qm02)--(qm22) --(qm20)--(qm00);
\fill [color=cyan,fill opacity=0.25](qm00)  --(qm02)--(qm22) --(qm20)--(qm00);
\draw[very thick, color=blue] (qm01) -- (qm11) --(qm21);
\draw[very thick, color=blue] (qm10) -- (qm11) --(qm12);

\fill (qm00) circle (3pt);
\fill (qm01) circle (3pt);
\fill (qm10) circle (3pt);
\fill (qm11) circle (3pt);
\fill (qm20) circle (3pt);
\fill (qm21) circle (3pt);
\fill (qm02) circle (3pt);
\fill (qm12) circle (3pt);
\fill (qm22) circle (3pt);

\coordinate (q00) at (1,1);
\coordinate (q01) at (1,2);
\coordinate (q10) at (2,1);
\coordinate (q11) at (2,2);
\coordinate (q20) at (1,3);
\coordinate (q21) at (2,3);
\coordinate (q02) at (3,1);
\coordinate (q12) at (3,2);
\coordinate (q22) at (3,3);

\node[below, fill=white,font=\fontsize{15pt}{12pt}\selectfont] at (-2,0.5) {support};

\draw [very thick, color=blue](q00)  --(q02)--(q22) --(q20)--(q00);
\fill [color=cyan,fill opacity=0.25](q00)  --(q02)--(q22) --(q20)--(q00);
\draw[line width =3pt, color=red] (q01) -- (q11) --(q21);
\draw[line width =3pt, color=red] (q10) -- (q11) --(q12);
\fill (q00) circle (3pt);
\fill (q01) circle (3pt);
\fill (q10) circle (3pt);
\fill (q11) circle (3pt);
\fill (q20) circle (3pt);
\fill (q21) circle (3pt);
\fill (q02) circle (3pt);
\fill (q12) circle (3pt);
\fill (q22) circle (3pt);

\node[below, fill=white,font=\fontsize{15pt}{12pt}\selectfont] at (2,0.5) {bulk};

\coordinate (q001) at (6,2);
\coordinate (q010) at (7,1);
\coordinate (q011) at (7,2);
\coordinate (q021) at (8,2);
\coordinate (q012) at (7,3);

\coordinate (q101) at (5,1);
\coordinate (q102) at (5,2);
\coordinate (q103) at (5,3);
\coordinate (q111) at (6,1);
\coordinate (q112) at (6,2);
\coordinate (q113) at (6,3);
\coordinate (q121) at (7,1);
\coordinate (q122) at (7,2);
\coordinate (q123) at (7,3);

\draw[very thick, color =blue] (q101) -- (q103) -- (q123) -- (q121) --(q101);
\fill [color=cyan,fill opacity=0.25](q101) -- (q103) -- (q123) -- (q121) --(q101);

\draw[very thick, color=blue] (q102) -- (q112) --(q122);
\draw[very thick, color=blue] (q111) -- (q112) --(q113);

\draw [line width =3pt, color=red](q001)--(q011) --(q021);
\draw [line width =3pt, color=red](q010)--(q011)--(q012);

\node[below, fill=white,font=\fontsize{15pt}{12pt}\selectfont] at (6,0.5) {boundary};

\fill (q010) circle (3pt);
\fill (q011) circle (3pt);
\fill (q021) circle (3pt);
\fill (q012) circle (3pt);
\fill (q101) circle (3pt);
\fill (q102) circle (3pt);
\fill (q103) circle (3pt);
\fill (q111) circle (3pt);
\fill (q112) circle (3pt);
\fill (q113) circle (3pt);

\end{tikzpicture}
}
\caption{Bulk and boundary Gauss's law operators. 
In this example, we consider a support composed of four plaquettes. 
Blue links indicate the support, and red links show the qubits on which Gauss's law operators act. The left is the support composed of four plaquettes. The center example represents the bulk Gauss's law, and the right is the boundary Gauss's law.}
\label{fig:bulk-gauss-law}
\end{figure}


\def\rscale{0.6}
\begin{table*}[!tb]
\centering
\caption{Pauli pool operator counts per support with and without the reductions. Blue vertices denote fully included Gauss's laws. Boundaries of the entire lattice are not taken into account.}
\label{tab:reduced-pauli-pool-z2lgt}
\tabcolsep = 0.7cm
\begin{tabular}{c|c||ccc}
\hline
\rule{0pt}{15pt}

\raisebox{3pt}{support $\mathcal{D}$} & \raisebox{3pt}{the maximum weight $D$} & \raisebox{3pt}{$|\mathcal{P}\setminus I^{\otimes {D}}|$} & \raisebox{3pt}{$|\mathcal{P}_{\mathcal{G},\text{odd}}|$}&\raisebox{3pt}{$|\mathcal{P}_{\mathcal{G},\text{odd}}/\mathcal{G}|$}\\
\hline
\rule{0pt}{25pt}
\raisebox{9pt}{\scalebox{\rscale}{

\begin{tikzpicture}[baseline={([yshift=-.8ex]current bounding box.center)}]
\coordinate (q00) at (0,0);\fill (q00) circle (3pt);
\coordinate (q01) at (0,1);\fill (q01) circle (3pt);
\coordinate (q10) at (1,0);\fill (q10) circle (3pt);
\coordinate (q11) at (1,1);\fill (q11) circle (3pt);

\draw [very thick](q00)  -- (q01)-- (q11) -- (q10)--(q00);
\end{tikzpicture}
}}&\raisebox{9pt}{$4$} & \raisebox{9pt}{$255$} & \raisebox{9pt}{$8$} & \raisebox{9pt}{$8$}\\\hline
\rule{0pt}{25pt}

\raisebox{9pt}{\scalebox{\rscale}{
\begin{tikzpicture}[baseline={([yshift=-.8ex]current bounding box.center)}]

\coordinate (q00) at (0,0);\fill (q00) circle (3pt);
\coordinate (q01) at (0,1);\fill (q01) circle (3pt);
\coordinate (q10) at (1,0);\fill (q10) circle (3pt);
\coordinate (q11) at (1,1);\fill (q11) circle (3pt);
\coordinate (q20) at (2,0);\fill (q20) circle (3pt);
\coordinate (q21) at (2,1);\fill (q21) circle (3pt);

\draw [very thick](q11) --(q01)--(q00) --(q10);
\draw [very thick](q10)--(q20) --(q21)--(q11);
\draw[very thick] {(q10)--(q11)};

\end{tikzpicture}
}}&\raisebox{9pt}{$7$} & \raisebox{9pt}{$16383$} & \raisebox{9pt}{$192$} & \raisebox{9pt}{$192$}\\\hline
\rule{0pt}{25pt}
\raisebox{9pt}{\scalebox{\rscale}{
\begin{tikzpicture}[baseline={([yshift=-.8ex]current bounding box.center)}]
\coordinate (q00) at (0,0);\fill (q00) circle (3pt);
\coordinate (q01) at (0,1);\fill (q01) circle (3pt);
\coordinate (q10) at (1,0);\fill (q10) circle (3pt);
\coordinate (q11) at (1,1);\fill (q11) circle (3pt);
\coordinate (q20) at (2,0);\fill (q20) circle (3pt);
\coordinate (q21) at (2,1);\fill (q21) circle (3pt);
\coordinate (q30) at (3,0);\fill (q30) circle (3pt);
\coordinate (q31) at (3,1);\fill (q31) circle (3pt);

\draw [very thick](q11)--(q01)--(q00) --(q10);
\draw [very thick](q20)--(q30) --(q31)--(q21);
\draw [very thick](q11)--(q21);
\draw [very thick](q10)--(q20);
\draw[very thick] {(q10)--(q11)};
\draw[very thick] {(q21)--(q20)};
\end{tikzpicture}
}}&\raisebox{9pt}{$10$}& \raisebox{9pt}{$1048575$} & \raisebox{9pt}{$2048$} & \raisebox{9pt}{$2048$} \\\hline
\rule{0pt}{45pt}

\raisebox{18pt}{\scalebox{\rscale}{
\begin{tikzpicture}[baseline={([yshift=-.8ex]current bounding box.center)}]

\coordinate (q00) at (0,0);\fill (q00) circle (3pt);
\coordinate (q01) at (0,1);\fill (q01) circle (3pt);
\coordinate (q10) at (1,0);\fill (q10) circle (3pt);
\coordinate (q11) at (1,1);\fill (q11) circle (3pt);
\coordinate (q20) at (2,0);\fill (q20) circle (3pt);
\coordinate (q21) at (2,1);\fill (q21) circle (3pt);
\coordinate (q02) at (0,2);\fill (q02) circle (3pt);
\coordinate (q12) at (1,2);\fill (q12) circle (3pt);

\draw [very thick](q01)--(q02) --(q12)--(q11);
\draw [very thick](q01)--(q00) --(q10);
\draw [very thick](q10)--(q20) --(q21)--(q11);
\draw[very thick] {(q10)--(q11)};
\draw[very thick] {(q01)--(q11)};
\fill[color=blue] (q11) circle (5pt);
\end{tikzpicture}
}}&\raisebox{18pt}{$10$}& \raisebox{18pt}{$1048575$} & \raisebox{18pt}{$2048$} & \raisebox{18pt}{$1024$} \\\hline
\rule{0pt}{45pt}

\raisebox{18pt}{\scalebox{\rscale}{
\begin{tikzpicture}[baseline={([yshift=-.8ex]current bounding box.center)}]

\coordinate (q00) at (0,0);\fill (q00) circle (3pt);
\coordinate (q01) at (0,1);\fill (q01) circle (3pt);
\coordinate (q10) at (1,0);\fill (q10) circle (3pt);
\coordinate (q11) at (1,1);\fill (q11) circle (3pt);
\coordinate (q20) at (2,0);\fill (q20) circle (3pt);
\coordinate (q02) at (0,2);\fill (q02) circle (3pt);
\coordinate (q21) at (2,1);\fill (q21) circle (3pt);
\coordinate (q12) at (1,2);\fill (q12) circle (3pt);
\coordinate (q22) at (2,2);\fill (q22) circle (3pt);

\draw [very thick](q12)--(q02)--(q01);
\draw [very thick](q01)--(q00)--(q10);
\draw [very thick](q10)--(q20)--(q21);
\draw [very thick](q21)--(q22)--(q12);
\draw[very thick] {(q11)--(q10)};
\draw[very thick] {(q11)--(q01)};
\draw[very thick] {(q11)--(q12)};
\draw[very thick] {(q11)--(q21)};
\fill[color=blue] (q11) circle (5pt);
\end{tikzpicture}
}}&\raisebox{18pt}{$12$}& \raisebox{18pt}{$16777215$} & \raisebox{18pt}{$30720$} & \raisebox{18pt}{$15360$} \\\hline
\end{tabular}

\end{table*}


\section{Numerical simulation}
\label{sec:numerical-simulation}
\subsection{Tensor network simulation settings}
To understand the algorithmic error and its dependence on the system and time-step size, we perform tensor network simulations of the QITE. Specifically, we use a matrix product state (MPS) where the systematic errors are controlled by the bond dimension and singular value decomposition (SVD) cutoff. We take the bond dimension large enough so that the systematic errors in the MPS representation are negligible, and the results can be used to study hardware algorithm errors in the absence of quantum noise. 
The gate operations in the QITE are implemented by the time evolving block decimation (TEBD) algorithm~\cite{Vidal:2003lvx}, and the coefficients~\eqref{eq:linear-eq-coeffs-QITE} are evaluated in MPS representations, then the linear equation~\eqref{eq:linear-eq-QITE}
is solved via a standard solver.
For comparison, we compare the QITE results (labeled \texttt{QITE}) with those from the following two methods.
The first is the Suzuki-Trotter decomposition of the imaginary-time evolution (labeled~\texttt{ITE})~\footnote{We distinguish ``(Q)ITE" from our numerical results written as \texttt{(Q)ITE}. The former is a theoretical (algorithmic) object while the latter is its numerical implementation using the specific methods.}.
This only involves the Suzuki-Trotter and TEBD errors and separates the unitary approximation error in the QITE from these errors.
The second is a density matrix renormalization group (DMRG)~\cite{PhysRevLett.69.2863}, which is the variational method using MPS to obtain the ground state energy. It provides a reference value for validating the \texttt{QITE} and \texttt{ITE} results~\footnote{To be more precise, we perform the DMRG study using the transverse field Ising model on the dual lattice, which is shown to be equivalent to the pure $\mathbb{Z}_2$ LGT~\cite{Wegner:1971app}.}.
For the simulation parameters, we set the SVD cutoff of TEBD and DMRG to $10^{-14}$, the number of DMRG sweeps to $5$, and the maximum bond dimension for DMRG to $200$.
Tensor network codes are implemented using a tensor network library~\texttt{Itensor.jl}~\cite{ITensor}.


\subsection{Results for $(N_{\text{site}}^x,N_{\text{site}}^y)=(3,3)$}

In this section, we present simulation results for a fixed system size $(N_{\text{site}}^x,N_{\text{site}}^y)=(3,3)$ (see the left panel of Fig.~\ref{fig:system_description}) with open boundary conditions. The initial state is the ground state of the electric field term given as~$\ket{\Omega_0} = \ket{+}^{\otimes N_\text{link}}$. We also fix the maximal imaginary time as $\tau_\text{max}=2.0$, for which we confirm the convergence (See the left panel of Fig.~\ref{fig:tau-epsilon}).

\begin{figure}[!t]
\centering
\scalebox{0.83}{
    \begin{tikzpicture}

\foreach \y in{0,1,2}{
\draw[very thick] (0,\y)--(2,\y);
\foreach \x in{0,...,2}{
\draw[very thick] (\x,0)--(\x,2);
\fill (\x,\y) circle (3pt);
}
}

\foreach \y in{0,1,2}{
\ifnum\y=0
\draw[line width =3pt] (4,\y)--(5.5,\y);
\draw[line width =3pt] (6.5,\y)--(8,\y);
\else
\ifnum\y=2
\draw[line width =3pt] (4,\y)--(5.5,\y);
\draw[line width =3pt] (6.5,\y)--(8,\y);
\else
\draw[very thick] (4,\y)--(5.5,\y);
\draw[very thick] (6.5,\y)--(8,\y);
\fi\fi
\foreach \x in{4,5,7,8}{
\ifnum\x=4
\draw[line width =3pt] (\x,0)--(\x,2);
\else
\ifnum\x=8
\draw[line width =3pt] (\x,0)--(\x,2);
\else
\draw[very thick] (\x,0)--(\x,2);
\fi\fi
\fill (\x,\y) circle (3pt);
}
}

\node[] at (6,1) {$\cdots$};

\node[left] at (-0.5,1) {$N_{\text{site}}^y=3$};
\node[below] at (1,-0.5) {$N_{\text{site}}^x=3$};

\node[below] at (6,-0.5) {$3\le N_{\text{site}}^x\le7$};



\draw [line width =3pt](0,0)--(0,2) --(2,2)--(2,0)--(0,0);


\end{tikzpicture}
}
\caption{System sizes studied in this work. We employ a ladder-like geometry with fixed $N_{\text{site}}^y=3$ and open boundary conditions. The smallest system, $(N_{\text{site}}^x,N_{\text{site}}^y)=(3,3)$, is shown on the left and the largest system, $(N_{\text{site}}^x,N_{\text{site}}^y)=(7,3)$, is shown on the right.}
\label{fig:system_description}
\end{figure}
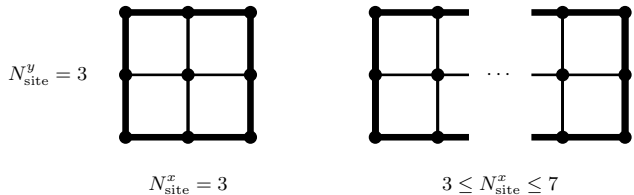

As we mentioned in Section~\ref{subsec:reduction-pauli-pool}, the size of the support of Pauli pools is shown to be controlled by the correlation length~\cite{Motta:2019yya}.
Although taking a sufficiently large support yields state preparation with provable error bounds, we can instead use a smaller support as a heuristic. This further reduces gate and measurement costs, at the expense of extra errors. 
Specifically, we consider Pauli pools at each step with the maximum Pauli-string weight of four (i.e. each Pauli string in the Pauli pool has weight at most four), which are given as follows (see also Tab.~\ref{fig:support-pauli-pool-z2lgt}). 
For a plaquette term~$W_{\bm{n}}$, we use $\mathcal{P}[\mathcal{D}]$ defined by Eq.~\eqref{eq:pauli-group-in-domain} with a support~$\mathcal{D}=\text{supp}(W_{\bm{n}})$, which is
a set of qubits that the plaquette acts on. 
On the other hand, each link~$(\bm{n},\bm{e}_{\mu})$ is involved in the two adjacent plaquettes, $W_{\bm{n}}$ and $W_{\bm{n}-\bm{e}_{\nu}}$, where~$\nu=y$~($\nu=x$) for~$\mu=x$~($\mu=y$).
We thus define
\begin{align}
    \mathcal{D}_{+} = \text{supp}(W_{\bm{n}})\,,\quad
    \mathcal{D}_{-} = \text{supp}(W_{\bm{n}-\bm{e}_{\nu}})\,,
\end{align}
and construct the corresponding Pauli pool as $\mathcal{P} = \mathcal{P}[\mathcal{D}_+] \cup \mathcal{P}[\mathcal{D}_-]$\,.

\begin{table}[!b]
\centering
\caption{The support of the Pauli pool for each term of the Hamiltonian. Blue links denote the links involved in the considered term.}
\label{fig:support-pauli-pool-z2lgt}
    \begin{tikzpicture}
\coordinate (q00) at (0,0);
\coordinate (q01) at (0,1);
\coordinate (q10) at (1,0);
\coordinate (q11) at (1,1);

\draw [ultra thick, color=blue](q00)  --(q01)--(q11) --(q10)--(q00);
\node[] at (0.5,0.5) {$W_{\bm{n}}$};
\node[below left] at (q00) {$\bm{n}$};
\fill (q00) circle (3pt);
\fill (q01) circle (3pt);
\fill (q10) circle (3pt);
\fill (q11) circle (3pt);

\coordinate (q20) at (4.0,0);
\coordinate (q21) at (4.0,1);
\coordinate (q30) at (5.0,0);
\coordinate (q31) at (5.0,1);
\draw [ultra thick, color=blue](q20)  --(q21)--(q31) --(q30)--(q20);
\node[below left] at (q20) {$\bm{n}$};

\fill (q20) circle (3pt);
\fill (q21) circle (3pt);
\fill (q30) circle (3pt);
\fill (q31) circle (3pt);

\draw[thin] (-0.5,-0.5) -- (7.5,-0.5);

\coordinate (q20) at (0,-2);
\coordinate (q30) at (1,-2);

\draw [very thick, color=blue](q20)  --node[above,color=black] {$X_{\bm{n},x}$}(q30);
\fill [fill=green, fill opacity=0.5](q20)--(q30);
\node[below left] at (q20) {$\bm{n}$};

\fill (q20) circle (3pt);
\fill (q30) circle (3pt);

\node[] at (4.65,-2) {and};

\coordinate (q120) at (3,-2);
\coordinate (q121) at (3,-1);
\coordinate (q130) at (4,-2);
\coordinate (q131) at (4,-1);
\draw [very thick, color=black](q120)  --(q121)--(q131) --(q130);
\draw [ultra thick, color=blue](q130)--(q120);

\node[below left] at (q120) {$\bm{n}$};
\fill (q120) circle (3pt);
\fill (q121) circle (3pt);
\fill (q130) circle (3pt);
\fill (q131) circle (3pt);

\coordinate (q220) at (5.5,-3);
\coordinate (q221) at (5.5,-2);
\coordinate (q230) at (6.5,-3);
\coordinate (q231) at (6.5,-2);
\draw [very thick, color=black](q221)--(q220)--(q230)--(q231);
\draw [ultra thick, color=blue](q221)--(q231);

\node[above left] at (q221) {$\bm{n}$};
\fill (q220) circle (3pt);
\fill (q221) circle (3pt);
\fill (q230) circle (3pt);
\fill (q231) circle (3pt);

\draw[thin] (-0.5,-3.5) -- (7.5,-3.5);

\coordinate (q40) at (0.5,-4);
\coordinate (q50) at (0.5,-5);

\draw [very thick, color=blue](q50)  --node[right,color=black] {$X_{\bm{n},y}$}(q40);
\fill [fill=green, fill opacity=0.5](q40)--(q50);
\node[below left] at (q50) {$\bm{n}$};
\fill (q40) circle (3pt);
\fill (q50) circle (3pt);

\node[] at (4.75,-4.5) {and};

\coordinate (q420) at (3,-5);
\coordinate (q421) at (3,-4);
\coordinate (q430) at (4,-5);
\coordinate (q431) at (4,-4);
\draw [very thick, color=black](q431)--(q421)--(q420)--(q430);
\draw [ultra thick, color=blue](q430)--(q431);

\node[below right] at (q430) {$\bm{n}$};
\fill (q420) circle (3pt);
\fill (q421) circle (3pt);
\fill (q430) circle (3pt);
\fill (q431) circle (3pt);

\coordinate (q520) at (5.5,-5);\fill (q520) circle (3pt);
\coordinate (q521) at (5.5,-4);\fill (q521) circle (3pt);
\coordinate (q530) at (6.5,-5);\fill (q530) circle (3pt);
\coordinate (q531) at (6.5,-4);\fill (q531) circle (3pt);
\draw [very thick, color=black](q521)--(q531)--(q530)--(q520);
\draw [ultra thick, color=blue](q520)--(q521);

\node[below left] at (q520) {$\bm{n}$};

\draw[thin] (2,2.25) -- (2,-5.5);
\draw[thin] (-0.5,-5.5) -- (7.5,-5.5);
\draw[thin] (-0.5,1.5) -- (7.5,1.5);
\draw[thin] (-0.5,2.25) -- (7.5,2.25);

\node [] at (0.5,1.85) {terms};
\node [] at (4.75,1.85) {Pauli pool supports};

\end{tikzpicture}

\end{table}

We implement both the QITE and (Suzuki-Trotterized) ITE using the TEBD, and compare results with the DMRG energy. 
At each time step, the relative errors \begin{align}
    \varepsilon(\tau) = \left|\frac{E(\tau) - E_\text{DMRG}}{E_\text{DMRG}}\right|
\end{align}
are computed.
The left panel of Fig.~\ref{fig:tau-epsilon} shows the relative errors $\varepsilon(\tau)$ as a function of $\tau$ ($0\leq \tau\leq \tau_{\text{max}}=2.0$) for a small coupling value~$\lambda=0.5$. 
The value of time steps is varied as $\Delta \tau \in \{0.05, 0.0125\}$.
For all the cases, one can observe that the relative errors converge to the value below $10^{-6}$ for $\tau=2.0$, which is further improved as the number of steps increases. 
Moreover, the \texttt{QITE} results agree with those from the \texttt{ITE}, which shows that the QITE algorithm can approximate the Suzuki-Trotterized ITE, even for a limited Pauli pool (the maximum Pauli-string weight of four), at least for the small coupling value and system size.

Next, the dependence on coupling values is studied, and the results are shown in the right panel of Fig~\ref{fig:tau-epsilon}.
We see that the accuracy gets worse with increasing the value of~$\lambda$ (except for a few points for $\Delta\tau = 0.0125$), as expected from the fact that we use the weak-coupling (electric) ground state as an initial state.
However, the relative errors for coupling values in $0.5\leq\lambda\leq 5.0$ are still below~$0.1\:\%$.
One can also observe that increasing the time steps improves the results for all coupling values.


\begin{figure*}[!t]
\begin{minipage}{0.48\linewidth}
\centering
\includegraphics[width=0.9\columnwidth]{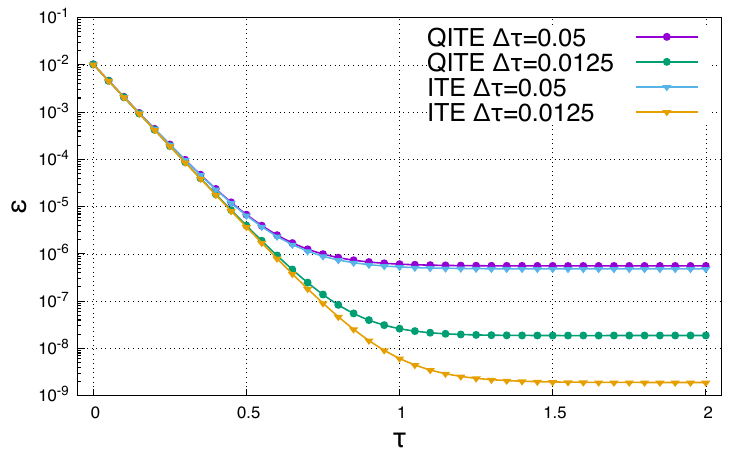}
\end{minipage}
\hspace{0.02\linewidth}
\begin{minipage}{0.48\linewidth}

\centering
\includegraphics[width=0.9\columnwidth]{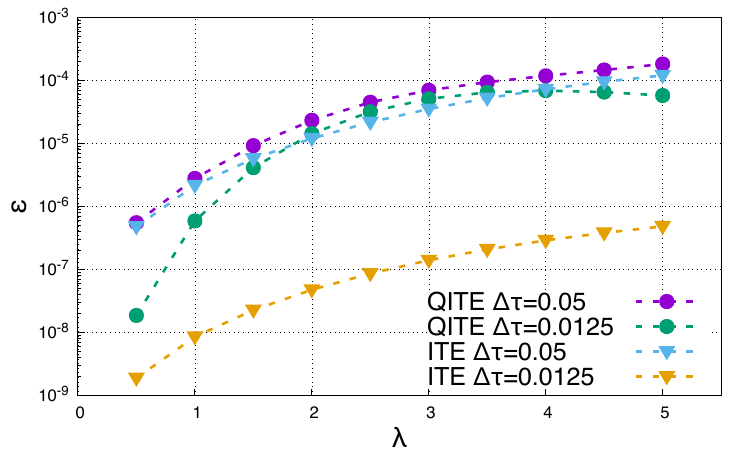}
\end{minipage}
\caption{Relative error of classically simulated QITE and ITE for $(N_{\text{site}}^x,N_{\text{site}}^y)=(3,3)$ with open boundary conditions. The initial state is the ground state of the electric term, and the Pauli pool is composed of the set of plaquettes. The left plot shows relative error as a function of imaginary time $\tau$ with the coupling value~$\lambda=0.5$ fixed. The right plot shows the dependence of the relative error on the coupling value~$\lambda$ for $\tau=2.0$.}

\label{fig:tau-epsilon}
\label{lambda-epsilon}

\end{figure*}


\subsection{Time step and system size dependence}

In this subsection, we present the dependence of the relative error on the time steps and system sizes. The simulation settings for larger systems are the same as those for~$(N_{\text{site}}^x,N_{\text{site}}^y)=(3,3)$.

First, we study the $\Delta \tau$ dependence for $\lambda=0.5, 2.0$ with system sizes $(N_{\text{site}}^x,N_{\text{site}}^y)=(3,3), (6,3)$ (see the right panel of Fig.~\ref{fig:system_description}). 
Fig. \ref{fig:dtau-epsilon} shows the relative error $\varepsilon$ at $\tau=\tau_\text{max}$ as a function of time step $\Delta\tau$. The values of time steps are taken as $\Delta \tau \in \{0.0125, 0.025, 0.05, 0.1\}$.
To interpret the results, we classify the possible sources of errors in the (Q)ITE algorithms as follows.
\begin{enumerate}[label=(\roman*),ref= (\roman*)]
  \item \label{itm:ite} Suzuki-Trotterized-ITE errors (\texttt{ITE} and \texttt{QITE})
  \begin{enumerate}[label=(\roman{enumi}-\alph*),ref=\roman{enumi}-\alph*]
    \item \label{itm:ite-finite-tau} finite $\tau_{\text{max}}$ effect
    \item \label{itm:ite-st} Suzuki-Trotter decomposition of ITE operator~\eqref{eq:ST-decomomp-1st}
  \end{enumerate}
  \item \label{itm:qite} QITE-specific errors (\texttt{QITE})
  \begin{enumerate}[label=(\roman{enumi}-\alph*),ref=(\roman{enumi}-\alph*)]
    \item \label{itm:qite-unitary-approx} unitary approximation of ITE operator~$e^{-\Delta\tau h^{(m)}}/\mathcal{N}\approx e^{-i\Delta\tau A}$ with a fixed Pauli pool
    \item \label{itm:qite-st} Suzuki-Trotter decomposition of unitary operator~$e^{-i\Delta\tau\tilde{A}}$
  \end{enumerate}
  \item \label{itm:mps} MPS errors (\texttt{ITE} and \texttt{QITE})
\end{enumerate}
The ITE~\ref{itm:ite} and MPS~\ref{itm:mps} errors are present both in \texttt{QITE} and \texttt{ITE}.
Therefore, one can consider that the difference between the \texttt{QITE} and \texttt{ITE} results comes from the QITE-specific error~\ref{itm:qite}.
We confirm that $\varepsilon$ for both \texttt{ITE} and \texttt{QITE} decreases when $\Delta\tau$ becomes smaller for $\lambda=0.5$.
This implies that not only the ITE error~\ref{itm:ite}, but also QITE-specific errors~\ref{itm:qite} are controlled by $\Delta\tau$ for a small coupling value ($\lambda=0.5$).
However, for $\lambda=2.0$, one can observe that ITE errors decrease as $\Delta\tau$ becomes smaller, while QITE errors saturate at some point~\footnote{
We also give additional results for a larger coupling value, $\lambda=3.5$, in Appendix~\ref{sec:additional-results}. Their qualitative behavior is similar to that observed at $\lambda=2.0$. The critical coupling value in the thermodynamic limit is $\lambda_c=3.04438$~\cite{Blote:2002ieo}. The difference between the small and large coupling behaviors observed here does not correspond sharply to this critical value. Nevertheless, finite-size effects can smooth and shift signatures associated with the thermodynamic-limit transition, and some of the behavior around $\lambda=2.0$ may be related to finite-size precursors of the phase transition. It would be interesting to clarify this relation using larger size simulation.
}.
Since the Suzuki-Trotter error in QITE~\ref{itm:qite-st} is considered to decrease with $\Delta\tau\to0$, we expect that this saturation of error mainly comes from unitary approximation error~\ref{itm:qite-unitary-approx}. This error can be mitigated by increasing the
maximum weight of Pauli strings in the Pauli pool  and/or using other initial states.
Moreover, the discrepancy between \texttt{ITE} and \texttt{QITE} for small $\Delta\tau$ becomes larger for $N_{\text{site}}=(6,3)$ than that for $N_{\text{site}}=(3,3)$. This is consistent with the argument above, since the effect of the Pauli weight is expected to increase with system size.


\begin{figure*}[!bt]
\begin{minipage}{0.48\linewidth}
\centering
\includegraphics[width=0.9\columnwidth]{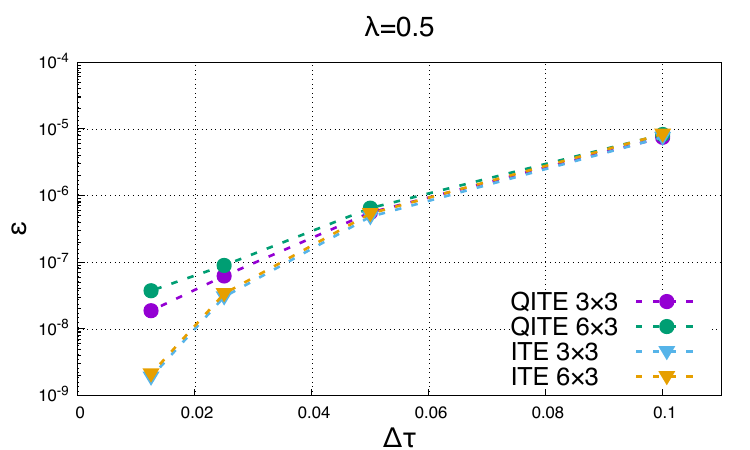}
\end{minipage}
\hspace{0.02\linewidth}
\begin{minipage}{0.48\linewidth}
\includegraphics[width=0.9\columnwidth]{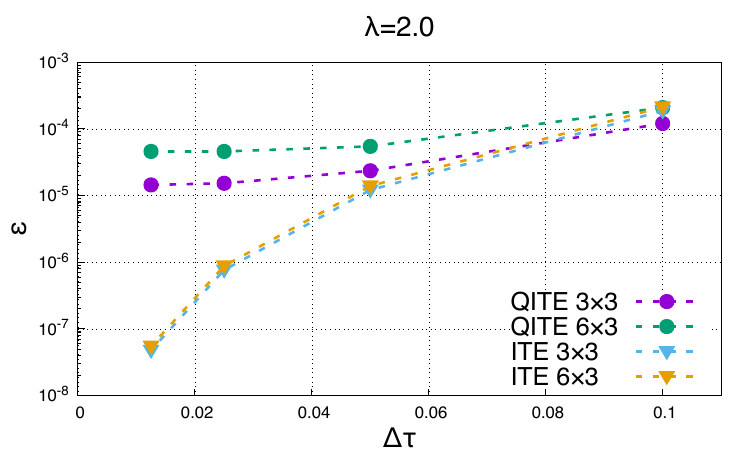}
\end{minipage}

\caption{Relative error of classically simulated QITE and ITE at the final imaginary time with respect to the DMRG energy as a function of imaginary time step. System size is $(N_{\text{site}}^x,N_{\text{site}}^y)=(3,3)$ to $(N_{\text{site}}^x,N_{\text{site}}^y)=(6,3)$, open boundary conditions. $\lambda =0.5$ for the left panel and $\lambda=2.0$ for the right panel, respectively. The initial state is the ground state of electric term. The QITE and ITE series correspond to results obtained using QITE and ITE, respectively.}

\label{fig:dtau-epsilon}

\end{figure*}



\begin{figure}[!t]
\centering
\includegraphics[width=0.98\columnwidth]{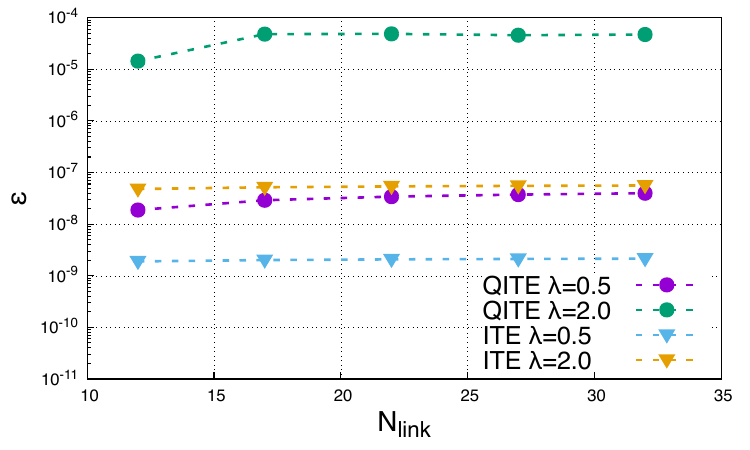}

\caption{Relative error of classically simulated QITE and ITE at the final imaginary time with respect to the DMRG energy as a function of the number of links. System size is $(N_{\text{site}}^x,N_{\text{site}}^y)=(3,3)$ to $(N_{\text{site}}^x,N_{\text{site}}^y)=(7,3)$, $\lambda =0.5$ and $\lambda =2.0$, open boundary conditions, imaginary time step $\Delta\tau=0.0125$. The initial state is the ground state of electric terms. The QITE and ITE series correspond to results obtained using QITE and ITE, respectively.}
\label{fig:system_size}
\end{figure}

We then investigate the system-size dependence in more detail.
Fig.~\ref{fig:system_size} shows the relative error $\varepsilon$ at $\tau=\tau_\text{max}$ as a function of the number of links (qubits), which allows us to see the system size dependence more clearly.
In this plot, the number of sites along the~$x$ axis varies as $N_\text{site}^x\in\{3,\cdots,7\}$ while the number of sites along the~$y$ axis is fixed at $N_\text{site}^y=3$ (see Fig.~\ref{fig:system_description}). The total number of links is therefore given by
\begin{align}
    N_\text{link} = N_\text{site}^x(N_\text{site}^y - 1) + (N_\text{site}^x - 1)N_\text{site}^y.
\end{align}
We see that the \texttt{ITE} errors remain almost constant, while the \texttt{QITE} errors slightly increase with system size, which implies that the QITE-specific error~\ref{itm:qite} becomes larger for larger $N_{\text{link}}$. The increasing rate is larger for $\lambda=2.0$ than that of $\lambda=0.5$, which is consistent with the previous discussion.
We also note that the \texttt{QITE} results for $\lambda=0.5$ show better performance than those for $\lambda=2.0$. 
In other words, the behavior observed in the right panel of Fig.~\ref{fig:tau-epsilon} for $(N^x_{\text{site}},N^y_{\text{site}})=(3,3)$ is also present for other system sizes.

\section{Conclusion}
\label{sec:conclusion}
In this work, we have applied the deterministic QITE algorithm and the resource reduction using symmetry to the $(2+1)$-dimensional pure $\mathbb{Z}_2$ lattice gauge theory and verified its performance via tensor-network simulations. 
First, we constructed the set of Pauli operators that commute with symmetry generators, which generalizes the previous results~\cite{wang2023symmetry}. This leads to a significant reduction of both measurement and gate costs using the general method in Ref.~\cite{Sun:2020dzq}.
For example, the number of elements in the Pauli pool with a single plaquette support is reduced from 255 to 8. The reduction factors are even more substantial in the remaining cases, and we can further combine the resource reduction using the quotient group.
Then, the numerical simulation was performed using TEBD. The QITE results are compared with the Suzuki-Trotterized ITE (without unitary approximation) and the DMRG. We confirmed that the energy obtained from the QITE agrees with the DMRG up to $0.1 \%$ error for the coupling value $0.5\leq \lambda \leq 5.0$ up to a twelve-plaquette system (two links in the vertical direction and six links in the horizontal direction). 
The error tends to increase with the value of $\lambda$, which is consistent with the fact that we take the electric ($\lambda=0$) ground state as an initial state of the ITE.
Moreover, we investigated the dependence of the error on the time step size and system size. 
We confirm that a shorter time step size~$\Delta\tau$ results in a smaller error for $\lambda=0.5$, which implies controllability of QITE errors. The error convergence is also observed for $\lambda=2.0$, but worse than that for $\lambda=0.5$. We discussed that this may come from the limited size of the Pauli pool support and the initial state. The validation of this hypothesis and the improvement are left for future work.
Finally, we saw that the QITE error slowly increases with the system size, at least the ladder-like geometry in $12\leq N_{\text{link}} (=N_{\text{q}})\leq32$. It would be interesting to see how the error scales for larger systems and/or genuine 2D systems.

Several future directions can be considered. 
First, understanding the dependence on the Pauli support size in a larger lattice would be important to see how the resources scale.
Additionally, in the large coupling regime, it would be more efficient to use other initial states, such as the ground state in the strong-coupling (magnetic) limit.
Moreover, to study the effects of noise, a noisy simulation or implementation on the current hardware is also valuable.
In particular, it would be important to see the robustness of symmetry preservation during the QITE in the presence of noise. 
For the gauge theory perspective, it would also be essential to include dynamical fermions. This makes the Gauss's law operators weight-five (Eq.~\eqref{eq:gauss-law} times fermion operator $X_{\bm{n}}$), and the number of Pauli operators in the symmetric pool is obtained by counting the gauge-invariant operators, whose~$Y/Z$ support ($\text{supp}(\bm{\sigma}_{\bar{X}})$) takes the form of closed Wilson loops or open Wilson lines terminated by pairs of fermion operators.
Finally, an extension to the LGTs with (continuous) Abelian and non-Abelian gauge symmetry would also be important for LGT simulations.

\begin{acknowledgments}
We thank Yutaro Iiyama for providing the codes for generating a two-dimensional lattice, based on which we performed our numerical simulation.
This work is supported by JST PRESTO Grant Number JPMJPR25F5 and JST Grant Number JPMJPF2221.
MS is supported by Forefront Physics and Mathematics Program to Drive Transformation (FoPM), a World-leading Innovative Graduate Study (WINGS) Program, the University of Tokyo.
\end{acknowledgments}

\appendix
\section{Derivation of the linear equation in QITE}
\label{sec:derivation-QITE}
First, the distance can be expanded in terms of $\Delta \tau$ as
\begin{align}
    \mathcal{D}(a_I^{(m)}) &=
    \|\frac{1}{\Delta \tau}\left(
    \ket{\bar{\psi}^{(m)}_{\Delta\tau}} - \ket{\bar{\psi}}
    \right)-iA^{(m)}\ket{\bar{\psi}}\|^2
    + \mathcal{O}(\Delta \tau)
    \\
    &=
    \|\ket{\delta\psi^{(m)}}
        - iA^{(m)}\ket{\bar{\psi}}\|^2
    + \mathcal{O}(\Delta \tau)\,,
\end{align}
where we defined 
\begin{align}
    \ket{\delta\psi^{(m)}} &:=
    \frac{1}{\Delta \tau}\left(
    \ket{\bar{\psi}^{(m)}_{\Delta\tau}} - \ket{\bar{\psi}}
    \right)
    \\
    &=
    \frac{1}{\Delta \tau}\left(
    \frac{1}{\mathcal{N}(\ket{\psi^{(m)}_{\Delta\tau}})}e^{-\Delta \tau h^{(m)}} - 1
    \right)\ket{\bar{\psi}}
    \\
    &=
    \frac{1}{\Delta \tau}
    \left[
        \frac{1-\Delta\tau h^{(m)}}{\sqrt{1-2\braket{\bar{\psi}|h^{(m)}|\bar{\psi}}}}-1
        \right]\ket{\bar{\psi}} + \mathcal{O}(\Delta \tau)\,.
\end{align}
The stationary condition leads to the following equations:
\begin{align}
    0 &= \frac{\partial}{\partial a_{I}}\mathcal{D}(\{a_J^{(m)}\})
    \\
    &=
    \frac{\partial}{\partial a_{I}}
    \left[
    -2\imag\bra{\bar{\psi}}A^{(m)\dag}\ket{\delta\psi^{(m)}}
    + \bra{\bar{\psi}}A^{(m)\dag}A^{(m)}\ket{\bar{\psi}}
    \right] 
    \\
    &=
    -2 \imag\bra{\bar{\psi}}\sigma_I^\dagger\ket{\delta\psi^{(m)}}
    + \sum_{J} a_J^{(m)}
    \left(
    2\real    
    \bra{\bar{\psi}}\sigma_I^\dagger\sigma_J\ket{\bar{\psi}}
    \right)\\
    &=-b_I +\sum_{J}S_{IJ}a_J
\end{align}
for all $I$, where
\begin{align}
\begin{cases}
        S_{IJ}=2\real\bra{\bar{\psi}}\sigma_I^\dagger\sigma_J\ket{\bar{\psi}}\,,\\
    b_I = 2 \imag\bra{\bar{\psi}}\sigma_I^\dagger\ket{\delta\psi^{(m)}}.
\end{cases}    
\end{align}

Therefore, an optimal tuple of $\bm{a}=\{a_I\}_I$ can be found by solving the linear equation
\begin{align}
    \bm{Sx}=\bm{b}.
\end{align}

Additionally, one can show the explicit form of $b_I$ up to the first order in $\Delta\tau$
\begin{align}
    b_I=-\frac{2\imag\bra{\bar\psi}\sigma_I^\dag h^{(m)}\ket{\bar\psi}}{\sqrt{1-2\Delta\tau\bra{\bar\psi}h^{(m)}\ket{\bar\psi}}}+\mathcal{O}(\Delta\tau^2)\,.
\end{align}

\section{Proof of Pauli pool reduction}
\label{app:proof-pauli-pool-reduction}
\begin{definition}
    For a pair of Pauli pools $\mathcal{P}$ and $\mathcal{P}^\prime$, we introduce a binary relation $\mathcal{P}\sim\mathcal{P}^\prime$ by
    \begin{align}
    \mathcal{P}\sim\mathcal{P}^\prime\quad\text{iff}\quad \ket{\tilde\Psi_{\mathcal{P}}^{(m)}}=\ket{\tilde\Psi_{\mathcal{P}^\prime}^{(m)}}.
    \end{align}
    This implies that both Pauli pools $\mathcal{P}$ and $\mathcal{P}^\prime$ give the same imaginary time evolution in the QITE algorithm.
\end{definition}

\begin{lemma}
    Suppose that there exists a subset $\mathcal{P}_{0}\subset \mathcal{P}$ that satisfies
\begin{align}\label{eq:reduction-condition}
    \begin{cases}
        S_{IJ}=0 \quad \text{for all} \quad\sigma_I\in \mathcal{P}_{0}, 
        \sigma_J\in \mathcal{P}\backslash\mathcal{P}_{0}\,,
        \\
        b_I=0  \quad\text{for all} \quad\sigma_I \in \mathcal{P}_{0}\,.
    \end{cases}
\end{align}
Then the time evolution obtained via the linear equation~\eqref{eq:linear-eq-QITE} from the full Pauli pool $\mathcal{P}$ is equivalent to that obtained from the reduced pool $\mathcal{P}_{\text{red}}:= \mathcal{P}\setminus\mathcal{P}_{0}$. In other words, we have
\begin{align}
\mathcal{P}\sim\mathcal{P}_{\text{red}}\,.
\end{align}
\label{lem:trivial}
\end{lemma}

\begin{proof}
Suppose that there exists $\mathcal{P}_{0}\subset\mathcal{P}$. Then the matrix $\bm S$ and the vector $\bm b$ have the form
\begin{align}
    \bm{S}&=
    \begin{pmatrix}
        \bm S^\prime & \bm 0 \\
        \bm 0 &  \bm *\\
    \end{pmatrix}\,,\\
    \bm b&=
    \begin{pmatrix}
        \bm b^\prime\\
        \bm 0\\
    \end{pmatrix}\,.
\end{align}
Therefore, 
a solution of the equation $\bm{S}'\bm{a}'=\bm b^\prime$ denoted by $\tilde{\bm{a}}'$
gives that of the equation  $\bm{S}\bm{a}=\bm b$ by the relation
\begin{align}
    \tilde{\bm{a}} = \begin{pmatrix}
        \tilde{\bm{a}}'\\
        \bm 0^\prime
    \end{pmatrix}\,.
\end{align}
Thus, the real-time evolution via $\mathcal{P}$ is equivalent to that of $\mathcal{P}_{\text{red}}$,
\begin{equation}
    e^{iA_\mathcal{P}(\tilde{\bm{a}})\Delta t}\ket{\psi}
    = e^{iA_\mathcal{P_{\text{red}}}(\tilde{\bm{a}}')\Delta t}\ket{\psi}\,.
\end{equation}
\end{proof}

\begin{proposition}
\label{parity-red}
Suppose that all elements of the matrix (vector) representation of the Hamiltonian and the initial state in the computational basis are real. Then
$\mathcal{P}_{\text{odd}}=
        \{\bm{\sigma}\in \mathcal{P} \mid \text{$\text{num}_Y(\bm{\sigma})$ is odd}\}$ is equivalent to the original Pauli pool $\mathcal{P}$
        \begin{align}
            \mathcal{P}\sim\mathcal{P}_{\text{odd}}\,.
        \end{align}
\end{proposition}

\begin{proof}
    From the previous lemma \ref{lem:trivial}, it is sufficient to prove that the condition \ref{eq:reduction-condition} holds for $\mathcal{P}_{\text{even}}$. It can be shown that for $\sigma_I\in\mathcal{P}_{\text{odd}}$ and $\sigma_J\in\mathcal{P}_{\text{even}}$, all elements of the matrix representation of $\sigma_I\sigma_J$ in the computational basis are purely imaginary. Thus, one can obtain
    \begin{align}
        S_{IJ}&=2\real{\bra{\psi}\sigma_I\sigma_J\ket{\psi}}=0\,.
    \end{align}
    The same argument applies to $b_I$, and hence one can use Lemma~\ref{lem:trivial}. The remaining part of the proof follows by induction over the imaginary time steps.
\end{proof}

\begin{proposition}
    \label{symmetry-red}
    Let us denote the model Hamiltonian as $H$ and a symmetry group as $\mathcal{S}$.
We then assume the following conditions.
\begin{itemize}
    \item (Symmetry of initial state) The initial state satisfies $s\ket{\psi_{\text{init}}}=\ket{\psi_{\text{init}}}$ for all~$s\in \mathcal{S}$.
    \item (Symmetry of Hamiltonian) The Hamiltonian satisfies $[H,s]=0$ for all $s \in \mathcal{S}$
    \item (Symmetric decomposition) Each term in the Suzuki-Trotter decomposition satisfies $[h^{(m)},s]=0$ for all $m=1,\ldots,M$ and $s\in \mathcal{S}$.
\end{itemize}
With these conditions, $\mathcal{P}_\mathcal{S} = \{\bm{\sigma}\in \mathcal{P} \mid \bm{\sigma} = s \bm{\sigma} s^\dagger, \forall s\in \mathcal{S}\}$ is equivalent to the original Pauli pool $\mathcal{P}$
\begin{align}
    \mathcal{P}\sim\mathcal{P}_{\mathcal{S}}\,.
\end{align}
\end{proposition}

\begin{proof}
        From Lemma~\ref{lem:trivial}, it is sufficient to prove that the condition \ref{eq:reduction-condition} holds for $\mathcal{P}\backslash\mathcal{P}_\mathcal{S}$.
        One can show that $\sigma\in\mathcal{P}\backslash\mathcal{P}_\mathcal{S}$ anticommutes with some $s\in\mathcal{S}$. Then, for $\sigma_I \in\mathcal{P}\backslash\mathcal{P}_\mathcal{S}$ and $\sigma_J\in\mathcal{P}_\mathcal{S}$
    \begin{align}
        S_{IJ}&=2\real{\bra{\psi}\sigma_I\sigma_J\ket{\psi}}\\
        &=2\real{\bra{\psi}\sigma_I\sigma_Js\ket{\psi}}\\
        &=2\real{\bra{\psi}(-s)\sigma_I\sigma_J\ket{\psi}}\\
        &= -S_{IJ}\,.
    \end{align}
    Thus, $S_{IJ}=0$. The same argument holds for $b_I$ and hence one can use lemma \ref{lem:trivial}. The remaining part of the proof is just using induction for each steps.
\end{proof}

\begin{proposition}
    Suppose that a Pauli pool $\mathcal{P}$ reduces to $\mathcal{P}_A$ and $\mathcal{P}_B$ according to Lemma.~\ref{lem:trivial}. Then the intersection of $\mathcal{P}_A$ and $\mathcal{P}_B$ is equivalent to the original Pauli pool $\mathcal{P}$
    \begin{align}
        \mathcal{P}\sim\mathcal{P}_A\cap\mathcal{P}_B
    \end{align}
\end{proposition}
\begin{proof}
    By assumption, there are subsets $\overline{\mathcal{P}_{A}}=\mathcal{P}\backslash\mathcal{P}_{A}$ and $\overline{\mathcal{P}_{B}}=\mathcal{P}\backslash\mathcal{P}_{B}$ for the case of $\mathcal{P}_{A}$ and $\mathcal{P}_{B}$, respectively, that satisfy the reduction condition~\ref{eq:reduction-condition} replacing $\mathcal{P}_0$ with $\overline{\mathcal{P}_{A}}$, $\overline{\mathcal{P}_{B}}$. 
    Then one can show that their union $\overline{\mathcal{P}_{A}}\cup\overline{\mathcal{P}_{B}}$ also satisfies the reduction condition~\ref{eq:reduction-condition}. 
    Thus, the complement of the union is equivalent to the original Pauli pool $\mathcal{P}$, i.e.,
    \begin{align}
        \mathcal{P}\backslash(\overline{\mathcal{P}_{A}}\cup\overline{\mathcal{P}_{B}})=\mathcal{P}_A\cap\mathcal{P}_B
    \end{align}
    where the equality follows from De Morgan's laws.
\end{proof}

\begin{corollary}
    Suppose that the conditions in Prop.~\ref{parity-red} and Prop.~\ref{symmetry-red} are satisfied. Then the intersection of $\mathcal{P}_\text{odd}$ and $\mathcal{P}_\mathcal{S}$ is equivalent to the original Pauli pool $\mathcal{P}$
    \begin{align}
        \mathcal{P}\sim\mathcal{P}_{\text{odd},\mathcal{S}}\,,
    \end{align}
    where we use $\mathcal{P}_{\text{odd},\mathcal{S}}=\mathcal{P}_{\text{odd}}\cap\mathcal{P}_{\mathcal{S}}$.
\end{corollary}

\begin{remark}
The following statement does not hold in general:
    \begin{align}
        \mathcal{P}\sim\mathcal{P}^\prime\quad \Rightarrow        \mathcal{P}\sim\mathcal{P}\cap\mathcal{P}^\prime\,.
    \end{align}
    Choosing different representatives in Prop.~\ref{prop:quotient} gives a counterexample.
\end{remark}

\begin{proposition}
\label{prop:quotient}
    The original Pauli pool $\mathcal{P}$ can be reduced to the intersection of the odd Pauli pool $\mathcal{P}_\text{odd}$ and the quotient Pauli pool of $S$-action $\mathcal{P}_{S}/S$
\begin{align}
\mathcal{P}_{\mathcal{S}}\sim\mathcal{P}_{\mathcal{S}}/\mathcal{S}\,.
\end{align}

\end{proposition}
\begin{proof}
Suppose that $\bm\sigma_k$ and $\bm\sigma_l$ belong to the same equivalence class. Then there exists some $s\in \mathcal{S}$ such that
\begin{align}
    \bm\sigma_k=\eta_{kl}\bm\sigma_l\cdot s,
\end{align}
where $\eta_{kl}=\pm1$. Thus,
\begin{align}
    \bra{\psi}\bm\sigma_k^\dagger h^{(m)}\ket{\psi} = \bra{\psi}s^\dagger\eta_{kl}\bm\sigma_l^\dagger h^{(m)}\ket{\psi}=\eta_{kl}\bra{\psi}\bm\sigma_l^\dagger h^{(m)}\ket{\psi},
\end{align}
and therefore
\begin{align}
b_{\bm\sigma_k}^{(m)}=\eta_{kl}b_{\bm\sigma_l}^{(m)}.
\end{align}
A similar argument can be applied to the matrix $S$ and we have
\begin{align}
    \bm{S}&=
    \begin{pmatrix}
        \bm S^\prime & \bm x_k & \eta_{kl} \bm x_k\\
        \bm x_k^\top & 2 & \eta_{kl}2\\
        \eta_{kl} \bm x_k^\top & \eta_{kl}2 & 2 \\
    \end{pmatrix}\,,\\
    \bm b&=
    \begin{pmatrix}
        \bm b^\prime\\
        b_k\\
        \eta_{kl}b_k\\
    \end{pmatrix}\,.
\end{align}
Now, let $\bm a$ be a solution to the linear equation $\bm{Sx}=\bm b$. Writing $\bm a^\top=\begin{pmatrix}
    \bm a^{\prime\top} & a_k & a_l
\end{pmatrix}$, one finds that 
\begin{align}
    \tilde{\bm a}=\begin{pmatrix}
        \bm a^\prime\\
        a_k+\eta_{kl}a_l\\
        0
    \end{pmatrix}
\end{align}
is also a solution to the same linear equation. Hence, instead of using the Pauli pool $\mathcal{P}_\mathcal{S}$, one can use $\mathcal{P}_\mathcal{S}\backslash\{\sigma_l\}$ with
\begin{align}
    \tilde{\bm{S}}&=
    \begin{pmatrix}
        \bm S^\prime & \bm x_k\\
        \bm x_k^\top & 2 \\
    \end{pmatrix}\,,\\
    \tilde{\bm b}&=
    \begin{pmatrix}
        \bm b^\prime\\
        b_k\\
    \end{pmatrix}\,,\\
    \tilde{\bm a}&=
    \begin{pmatrix}
        \bm a^\prime\\
        a_k+\eta_{kl}a_l\\
    \end{pmatrix}\,,
\end{align}
which produces the same result. Therefore $\mathcal{P}_\mathcal{S}\sim\mathcal{P}_\mathcal{S}\backslash\{\sigma_l\}$. Repeating this argument for all equivalent classes and for all non-representative elements in each class, we have 
\begin{align}
    \mathcal{P}_\mathcal{S}\sim\mathcal{P}_\mathcal{S}/\mathcal{S}.
\end{align}
\end{proof}

\begin{corollary}
    Suppose that those conditions in the previous propositions are satisfied. Then the intersection of $\mathcal{P}_\text{odd}$ and $\mathcal{P}_\mathcal{S}/\mathcal{S}$ is equivalent to the original Pauli pool $\mathcal{P}$
    \begin{align}
        \mathcal{P}\sim\mathcal{P}_{\text{odd},\mathcal{S}}/\mathcal{S}\,.
    \end{align}
\end{corollary}

\section{Reduced Pauli pool in $\mathbb{Z}_2$ LGT}
\label{app:pauli-pool-proof-z2lgt}
In this section, we provide the proofs of propositions regarding the cost reduction in $\mathbb{Z}_2$  LGT, which are given in Section~\ref{subsec:qite-and-reduction-z2lgt}.

\begin{proposition}
    
    The number of elements in $\mathcal{P}_{\mathcal{G},\text{odd}}$ is given by
    \begin{align}
    |\mathcal{P}_{\mathcal{G},\text{odd}}|=2^{n_\text{link}(\mathcal{D})-1}\times(2^{n_\text{plaq}(\mathcal{D})}-1).
\end{align}
\end{proposition}

\begin{proof}
The Pauli pool $\mathcal{P}_{\mathcal{G},\text{odd}}$ is given by 
\begin{align}
        \mathcal{P}_{\mathcal{G},\text{odd}} &:= \mathcal{P}_{\mathcal{G}} \cap \mathcal{P}_{\text{odd}}
        \\
        &=
        \{\bm{\sigma}\in\mathcal{P}\mid \text{supp}(\bm{\sigma}_{\bar{X}}) \in \Gamma,
        \notag\\ &\qquad
        \land \text{num}_Y(\bm{\sigma}_{\bar{X}}) \equiv 1 \pmod 2\}\,.
        \label{eq:pauli-g-odd-condition}
\end{align}
Pauli strings $\bm{\sigma}$ for a fixed support $\mathcal{D}=\text{supp}(\bm{\sigma})$ satisfying the above condition can be constructed as follows (see Fig.~\ref{fig:assignment} for an example):
\begin{enumerate}
    \item Partition the support as $\mathcal{D}=\mathcal{D}_X \sqcup \mathcal{D}_{\bar{X}}$, where $\mathcal{D}_X := \text{supp}(\bm{\sigma}_X)$ and $\mathcal{D}_{\bar{X}} := \text{supp}(\bm{\sigma}_{\bar{X}})$.  
    \item Assign $I,X$ and $Y,Z$ for each qubit $q\in\mathcal{D}_X$ and $q\in\mathcal{D}_{\bar{X}}$, respectively.
\end{enumerate}
If one chooses the partition in the first step so that it satisfies the first condition in~\eqref{eq:pauli-g-odd-condition}, each choice corresponds to choosing a set of plaquettes, which specifies $\mathcal{D}_{\bar{X}}\in \Gamma$.
The number of such sets is given by
\begin{align}
    2^{n_\text{plaq}(\mathcal{D})}-1\,,
\end{align}
where we exclude the configuration with $|\mathcal{D}_{\bar{X}}|=0$, since this breaks the second condition $\text{num}_Y(\bm{\sigma}_{\bar{X}})$.

\begin{figure}[!t]
\centering
\scalebox{0.8}{
\begin{tikzpicture}

\coordinate (qm00) at (-3,1);
\coordinate (qm01) at (-3,2);
\coordinate (qm10) at (-2,1);
\coordinate (qm11) at (-2,2);
\coordinate (qm20) at (-3,3);
\coordinate (qm21) at (-2,3);
\coordinate (qm02) at (-1,1);
\coordinate (qm12) at (-1,2);
\coordinate (qm22) at (-1,3);

\draw [very thick, color=blue](qm00)  --(qm02)--(qm22) --(qm20)--(qm00);
\draw[very thick, color=blue] (qm01) -- (qm11) --(qm21);
\draw[very thick, color=blue] (qm10) -- (qm11) --(qm12);

\fill (qm00) circle (3pt);
\fill (qm01) circle (3pt);
\fill (qm10) circle (3pt);
\fill (qm11) circle (3pt);
\fill (qm20) circle (3pt);
\fill (qm21) circle (3pt);
\fill (qm02) circle (3pt);
\fill (qm12) circle (3pt);
\fill (qm22) circle (3pt);

\node[] at (-2,0.5){support $\mathcal{D}$};

\coordinate (q00) at (1,1);
\coordinate (q01) at (1,2);
\coordinate (q10) at (2,1);
\coordinate (q11) at (2,2);
\coordinate (q20) at (1,3);
\coordinate (q21) at (2,3);
\coordinate (q02) at (3,1);
\coordinate (q12) at (3,2);
\coordinate (q22) at (3,3);

\draw [very thick, color=blue](q00)  --(q02)--(q22) --(q20)--(q00);
\draw[line width =2pt, color=green,fill=green,fill opacity=0.20] (q01) -- (q11) --(q21)--(q20)--(q01);
\draw[line width =2pt, color=green,fill=green,fill opacity=0.20] (q10) -- (q11) --(q12)--(q02)--(q10);

\draw[line width =2pt, color=magenta] (q10) -- (q00) --(q01);
\draw[line width =2pt, color=magenta] (q12) -- (q22) --(q21);

\fill (q00) circle (3pt);
\fill (q01) circle (3pt);
\fill (q10) circle (3pt);
\fill (q11) circle (3pt);
\fill (q20) circle (3pt);
\fill (q21) circle (3pt);
\fill (q02) circle (3pt);
\fill (q12) circle (3pt);
\fill (q22) circle (3pt);

\draw[ultra thick,->] (-0.5,2)--(0.5,2);
\node[] at (2,0.5){partition $\mathcal{D}=\mathcal{D}_X \sqcup \mathcal{D}_{\bar{X}}$};

\coordinate (q101) at (5,1);
\coordinate (q102) at (5,2);
\coordinate (q103) at (5,3);
\coordinate (q111) at (6,1);
\coordinate (q112) at (6,2);
\coordinate (q113) at (6,3);
\coordinate (q121) at (7,1);
\coordinate (q122) at (7,2);
\coordinate (q123) at (7,3);

\draw[very thick, color =blue] (q101) -- (q103) -- (q123) -- (q121) --(q101);

\draw[very thick, color=blue] (q102) -- (q112) --(q122);
\draw[very thick, color=blue] (q111) -- (q112) --(q113);

\draw[line width =2pt, color=green,fill=green,fill opacity=0.20] (q102) --node[color=black]{$Y$} (q112) --(q113)--(q103)--(q102);
\draw[line width =2pt, color=green,fill=green,fill opacity=0.20] (q112) -- (q111) --(q121)--(q122)--(q112);

\draw[line width =2pt, color=magenta] (q113) -- (q123) --(q122);
\draw[line width =2pt, color=magenta] (q102) -- (q101) --(q111);

\node[] at (5.5,2) {$Y$};
\node[] at (6,2.5) {$Y$};
\node[] at (5.5,3) {$Y$};
\node[] at (5,2.5) {$Z$};

\node[] at (6.5,2) {$Z$};
\node[] at (6,1.5) {$Z$};
\node[] at (6.5,1) {$Y$};
\node[] at (7,1.5) {$Y$};

\node[] at (5,1.5) {$X$};
\node[] at (5.5,1) {$I$};
\node[] at (6.5,3) {$X$};
\node[] at (7,2.5) {$X$};

\fill (q101) circle (3pt);
\fill (q102) circle (3pt);
\fill (q103) circle (3pt);
\fill (q111) circle (3pt);
\fill (q112) circle (3pt);
\fill (q113) circle (3pt);
\fill (q121) circle (3pt);
\fill (q122) circle (3pt);
\fill (q123) circle (3pt);

\draw[ultra thick,->] (3.5,2)--(4.5,2);

\node[align=center] at (6,0.25){Pauli assignment \\ for $\mathcal{D}_{X}$ and $\mathcal{D}_{\bar{X}}$};

\end{tikzpicture}
}
\caption{An example of a partition and Pauli assignment for a four plaquette support~$\mathcal{D}$, shown as blue links. 
Green links denote a specific partition~$\mathcal{D}_{\bar{X}} = \text {supp}(\bm{\sigma}_{\bar X})\in \Gamma$ and green plaquettes denote the corresponding plaquette excitation. Magenta links show the support for the other part, $ \mathcal{D}_{{X}} = \text {supp}(\bm{\sigma}_{ X})=\mathcal{D}\setminus \text {supp}(\bm{\sigma}_{\bar X})$.}
\label{fig:assignment}
\end{figure}

Once $\mathcal{D}_X$ and $\mathcal{D}_{\bar{X}}$ are fixed, we assign $\{I,X\}$ or $\{Y,Z\}$ to each qubit in the supports.
Let us denote the weights of the Pauli strings as
\begin{align}
    w_{X} &:= \text{wt}(\bm{\sigma}_X) =|\mathcal{D}_X|\,,
    \\
    w_{\bar{X}} &:= \text{wt}(\bm{\sigma}_{\bar{X}}) =|\mathcal{D}_{\bar{X}}|\,.
\end{align}
The Pauli string $\bm{\sigma}_X$ can be taken from $\{I,X\}^{\otimes w_X}$. The number of choices is given by $2^{w_X}$.
On the other hand, 
$\bm{\sigma}_{\bar{X}}$ can be fixed with the condition $\text{num}_Y(\bm{\sigma}_{\bar{X}})$ being satisfied.
The number of possible strings is given as~\footnote{
The equality can be shown as follows. First, we have
\begin{align}
    2^n=&(1+1)^n =
    \sum_{k=0}^n \binom{n}{k}\,,
    \\
    0^n=&(1-1)^n= \sum_{k=0}^n\binom{n}{k} (-1)^k\,.
\end{align}
Therefore, we obtain
\begin{align}
    \sum_{\substack{k=1,\\k\equiv 1\:(\text{mod}\:2)}}^{n}
    \binom{n}{k}
    &=\sum_{k=0}^{n}
    \binom{n}{k}\frac{1-(-1)^k}{2}
    \\
    &= \frac{2^n - 0^n}{2} = 2^{n-1}
\end{align}
}
\begin{align}
    \sum_{\substack{k=1,\\k\equiv 1\:(\text{mod}\:2)}}^{w_{\bar{X}}}
    \begin{pmatrix}
        w_{\bar{X}}\\
        k
    \end{pmatrix}
    =2^{w_{\bar{X}}-1}\,.
\end{align}

Thus, for each partition $\mathcal{D}=\mathcal{D}_X\sqcup\mathcal{D}_{\bar{X}}$, the number of possible Pauli assignments is given by
\begin{align}
    2^{w_X} \times 2^{w_{\bar{X}}-1} 
\end{align}
On the other hand, since $\mathcal{D}$ is a disjoint union of $\mathcal{D}_X$ and $\mathcal{D}_{\bar{X}}$, we have
\begin{equation}
    w_X + w_{\bar{X}} = |\mathcal{D}| = n_{\text{link}}(\mathcal{D})\,,
\end{equation}
independent of the way of partitioning.
Thus, the number of possible Pauli strings for each partition is given by
\begin{align}
    2^{n_\text{link}(\mathcal{D})-1}\,.
\end{align}
Therefore, the total number of Pauli strings is obtained by multiplying the number of possible partitions by the above factor as
\begin{equation}
    2^{n_\text{link}(\mathcal{D})-1}\times(2^{n_\text{plaq}(\mathcal{D})}-1)\,.
\end{equation}
\end{proof}

\begin{proposition}
The number of elements in $\mathcal{P}_{\mathcal{G},\text{odd}}/\mathcal{G}$ is given by
\begin{align}
|\mathcal{P}_{\mathcal{G},\text{odd}}/\mathcal{G}|=\frac{2^{n_\text{link}(\mathcal{D})-1}\times(2^{n_\text{plaq}(\mathcal{D})}-1)}{2^{n_{\mathcal{G}}(\mathcal{D})}}\,.
\end{align}
\end{proposition}

\begin{proof}
Suppose $g_n$ is a bulk Gauss's law operator in the support of a Pauli pool $\mathcal{P}$. Then the right action of $g_n$ on $\mathcal{P}$ is given by the right multiplication of $X$ on each Pauli matrix
\begin{align}
    \begin{cases}
        IX = X\\
        XX = I\\
        YX =-iZ\\
        ZX =iY\\
    \end{cases}\,.
    \label{eq:x-action}
\end{align}
This means that the action of $g_n$ does not mix $\bm{\sigma}_X$ and $\bm{\sigma}_{\bar X}$, $\bm{\sigma}=\bm{\sigma}_X\otimes{\bm{\sigma}_{\bar X}}\to\pm \bm{\sigma}'_X\otimes{\bm{\sigma}'_{\bar X}}$ with $\text{supp}(\bm{\sigma}_X)$ and $\text{supp}(\bm{\sigma}_{\bar{X}})$ unchanged.
Consider $\bm{\sigma}=\bm{\sigma}_X\otimes{\bm{\sigma}_{\bar X}}\in\mathcal{P}_{\mathcal{G},\text{odd}}$, which thus satisfies the conditions~\eqref{eq:pauli-g-odd-condition}.
Since all closed loops on the square lattice always have an even number of links, we have
\begin{align}
    \text{num}_Y(\sigma_{\bar X})\equiv \text{num}_Z(\sigma_{\bar X}) \mod 2.
\end{align}
Thus, the exchange of $Z$ and $Y$ induced by the action~\eqref{eq:x-action} does not change the parity condition~$\text{num}_Y(\sigma_{\bar X})$. Hence, one finds 
\begin{align}
    \bm{\sigma}\cdot g_n \in\mathcal{P}_{\mathcal{G},\text{odd}}.
\end{align}
Moreover, we also have
\begin{align}
    (\bm{\sigma}\cdot g_n)\cdot g_n =\bm{\sigma}\cdot g_n^2=\bm{\sigma}\,.
\end{align}

Therefore, the right action of $g_n$ induces a free $\mathbb{Z}_2$ action on $\mathcal{P}$. Consequently,
\begin{align}
    |\mathcal{P}_{\mathcal{G},\text{odd}}/\{g_n\}|=\frac{|\mathcal{P}_{\mathcal{G},\text{odd}}|}{2}\,,
\end{align}
for a fixed $g_n$.
Repeating the same argument for all bulk Gauss's laws $g_n\in\mathcal{G}$, we have
\begin{align}
|\mathcal{P}_{\mathcal{G},\text{odd}}/\mathcal{G}|=\frac{|\mathcal{P}_{\mathcal{G},\text{odd}}|}{2^{n_{\mathcal{G}}(\mathcal{D})}}\,,
\end{align}
provided that these Gauss's law operators are independent.
\end{proof}

\section{Additional numerical results}
\label{sec:additional-results}

We present additional numerical results for the larger coupling value $\lambda=3.5$, which is above the thermodynamic-limit critical coupling $\lambda_c=3.04438$~\cite{Blote:2002ieo}.
Figure~\ref{fig:tau-epsilon-large-lambda} shows the imaginary-time dependence of the relative error for $\lambda=3.5$, corresponding to Fig.~\ref{fig:tau-epsilon} in the main text for $\lambda=0.5$.
Figure~\ref{fig:dtau-epsilon-large-lambda} shows the dependence of the final relative error on the imaginary-time step $\Delta\tau$, corresponding to Fig.~\ref{fig:dtau-epsilon} in the main text for $\lambda=0.5$ and $2.0$.
The qualitative behavior is similar to that observed at~$\lambda=2.0$.

\begin{figure}[h]
\centering
\includegraphics[width=0.9\columnwidth]{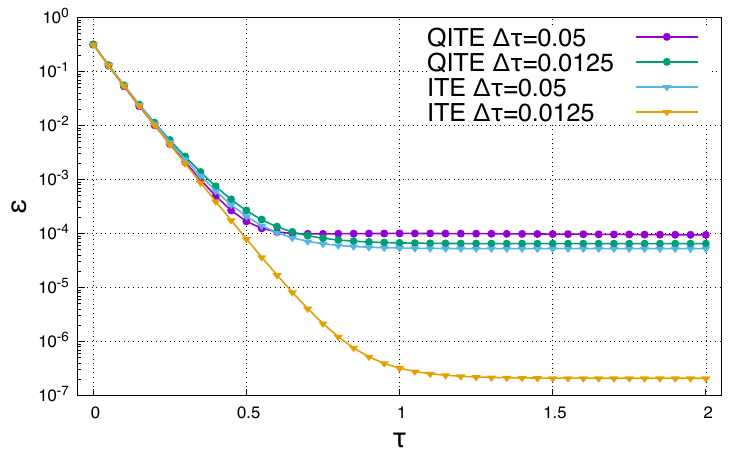}

\caption{Relative error of classically simulated QITE and ITE for $(N_{\text{site}}^x,N_{\text{site}}^y)=(3,3)$, and $\lambda=3.5$, with open boundary conditions. The initial state is the ground state of the electric term, and the Pauli pool is composed of the set of plaquettes. The plot shows relative error as a function of imaginary time $\tau$ with the coupling value~$\lambda=3.5$ fixed.}
\label{fig:tau-epsilon-large-lambda}
\end{figure}

\begin{figure}[h]
\centering
\includegraphics[width=0.9\columnwidth]{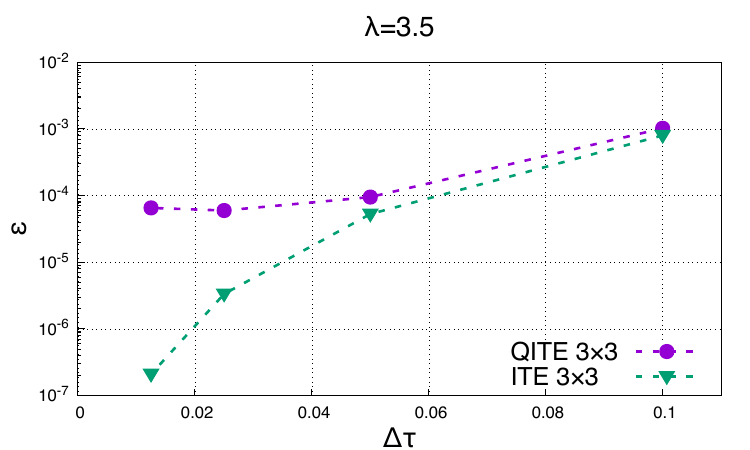}

\caption{Relative error of classically simulated QITE and ITE at the final imaginary time with respect to the DMRG energy as a function of imaginary time step. System size is $(N_{\text{site}}^x,N_{\text{site}}^y)=(3,3)$, and $\lambda =3.5$, with open boundary conditions. The initial state is the ground state of the electric term. The QITE and ITE series correspond to results obtained using QITE and ITE, respectively.}
\label{fig:dtau-epsilon-large-lambda}
\end{figure}

\clearpage

\bibliography{refs}
\end{document}